\documentclass{lmcs} 


\usepackage{hyperref}
\theoremstyle{plain} 

\usepackage{lscape}
\usepackage{amsmath}
\usepackage{amsthm}
\usepackage{amssymb}
\usepackage{latexsym}
\usepackage{xcolor}
\usepackage{proof}
\usepackage{bussproofs}
\usepackage{amssymb}
\usepackage{amsmath}
\usepackage{mathrsfs}
\usepackage{color}
\usepackage{cmll}
\usepackage{bbold}
\usepackage{tikz}
\usepackage{tikz-qtree}

\newtheorem{theorem}{Theorem}[section]
\newtheorem{lemma}[theorem]{Lemma}
\newtheorem*{theorem*}{Theorem}
\newtheorem{proposition}[theorem]{Proposition}

\theoremstyle{definition}
\newtheorem{definition}{Definition}[section]
\newtheorem{example}{Example}[section]
\theoremstyle{remark}

\newcommand{\impl}{\rightarrow}

\newcommand{\cho}{\bigoplus}

\newcommand{\bool}{\mathtt{Bool}}

\newcommand{\encodes}{\rightsquigarrow}

\newcommand{\thr}{\varepsilon}
\newcommand{\calpe}{\mathcal{P}\thr}

\newcommand{\lam}{\lambda}
\newcommand{\lan}{\langle}
\newcommand{\ran}{\rangle}

\newcommand{\test}{\circlearrowright}
\newcommand{\trust}{\mathtt{Trust}?}
\newcommand{\tert}{\mathtt{True}}
\newcommand{\terf}{\mathtt{False}}

\newcommand{\hcoin}{\mathtt{h}}
\newcommand{\tcoin}{\mathtt{t}}

\newcommand{\lamt}{$\lambda^{\test}_{\trust}$}

\newcommand{\calp}{\mathcal{P}}

\newcommand{\cals}{\mathcal{S}}
\newcommand{\calt}{\mathcal{T}}

\newcommand{\probab}{\mathrm{Pr}}
\newcommand{\conf}{\mathrm{Confidence}}

\newcommand{\dist}{\mathrm{D}_{\mathrm{TV}}}

\newcommand{\FV}{\mathrm{FV}}

\newcommand{\red}{\mapsto}

\begin{document}


\title[A Typed $\lambda$-Calculus for Establishing Trust in Probabilistic Programs]{A Typed $\lambda$-Calculus for Establishing Trust in Probabilistic Programs}

\thanks{This work was supported by the Ministero dell'Università e della Ricerca through the PRIN project BRIO (Bias, Risk and Opacity in AI), grant number 2020SSKZ7R.}	

\author[F.A.~Genco]{Francesco~A. Genco}[a]
\author[G.~Primiero]{Giuseppe Primiero}[b]

\address{Center for Logic, Language and Cognition (LLC)\\
Department of Philosophy and Education Sciences,
University of Turin}	
\email{francesco.genco@unito.it}  

\address{Logic, Uncertainty, Computation and Information Lab (LUCI)\\
Department of Philosophy, University of Milan}	
\email{giuseppe.primiero@unimi.it}  





\begin{abstract}
  \noindent The extensive deployment of probabilistic algorithms has radically changed our perspective on several well-established computational notions. {\it Correctness} is probably the most basic one. While a typical probabilistic program cannot be said to compute {\it the} correct result, we often have quite strong expectations about the frequency with which it should return certain outputs. In these cases, {\it trust} as a generalisation of correctness fares better. One way to understand it is to say that a probabilistic computational process is trustworthy if the frequency of its outputs is compliant with a probability distribution which models its expected behaviour. We present a formal computational framework that formalises this idea. In order to do so, we define the typed $\lam$-calculus \lamt{} which features operators for conducting experiments at runtime on probabilistic programs and for evaluating whether they compute outputs as determined by a target probability distribution. After proving some fundamental  computational properties of the calculus, such as progress and termination, we define a notion of {\it confidence} that allows us to prove that our notion of {\it trust} behaves correctly with respect to the basic tenets of probability theory.
\end{abstract}

\maketitle

\section{Introduction}

The extensive deployment of software systems based on non-deterministic algorithms to solve tasks which would be of difficult, if not impossible, solution by symbolic methods only, has radically changed our perspective on several well-established computational notions. Even some very elementary notions require revision if they are still to be fruitfully employed as we have done so far. 

Among them, the very basic notion of {\it correctness} is an extreme example of such phenomenon. On the one hand, obviously a typical probabilistic program cannot be said to compute {\it the} correct result, nor that it computes the expected result given any particular execution, since such a program is typically non-deterministic and thus expected to yield possibly different results given the same input. On the other hand, probabilistic programs employed in practice often induce quite strong expectations about the correct frequency with which certain outputs should be returned, and it would be rather worrying if no correlation could be found between the frequencies of different outputs and the variety of inputs given to the program. For instance: given a run of an image classifier on what is considered a decent amount of random pictures, a classification as a golden baroque staircase of most inputs would make one think the software is bogus. 

A generalisation of correctness that can be also applied to these cases seems to emerge rather naturally if one considers several examples of this kind. In the literature on philosophy of computation, a term strictly related to this intuitive notion is being used more and more, namely  {\it trust}. As we can trust someone to act according to our expectations even if we do not have a precise knowledge of what that person will exactly do, we can trust a probabilistic program even though we do not have precise knowledge of how the execution of that program will turn out -- see \cite{sep-trust,sim20} for a survey on the philosophical notion of trust and \cite{md05,cof07,kra15,key16,ess20,gmw20,sul20,cof21} for surveys and references concerning the usage of the notion of trust in computer science. While with respect to human behaviour this state of ignorance is determined by what we could naively call the complexity of life, in the context of probabilistic computation the lack of knowledge might be simply motivated by the non-determinism of the program, but it might also involve a certain amount of opacity concerning the mechanisms governing its behaviour \cite{esc21}. Since the definition of trust in general and in particular with respect to non-sentient agents is still  very problematic and far from being definitely settled, we restrict our attention to the notion of trustworthiness, that is, the property that makes an agent worth (or not) of our trust \cite{pk16}. In other words, we aim at capturing those objective features of the behaviour of a computation that constitute a reasonable ground for concluding that we trust it to accomplish a given task. For the sake of terminological simplicity, though, we will simply conflate the notions of  {\it trust} and {\it trustworthiness} in the rest of this paper.

The core intuition motivating our formal treatment is expressed by the following question: how can a program be considered trustworthy if one ignores how its executions will unfold? Maybe similarly to the case of trust among humans, one very simple answer naturally occurs: by experiencing that the behaviour of the program meets our expectations for a sufficient number of times. A probabilistic computational process is trustworthy if the frequency of its outputs is compliant with a probability distribution which models its expected behavior. The aim of the present work is then to present a formal computational framework in which it is possible to directly formalise a notion of trustworthiness for probabilistic computational processes according to this simple idea. 

We will do this by defining the calculus \lamt{},  a typed probabilistic $\lam$-calculus featuring operators for conducting experiments on probabilistic programs and for evaluating whether they compute outputs as determined by a target probability distribution, where the latter models the task one expect them to accomplish. Since this kind of evaluation does not rely in any way on the definition of the probabilistic program under consideration, rather it exclusively exploits its outputs at runtime, the operators of the calculus are well suited also for analysing the behaviour of opaque programs. For the sake of simplicity, we leave aside for the moment the problem of formally specifying how to represent blackbox programs in the calculus, but we stress that there seems to exist no relevant obstacle to this end. 

The probabilistic $\lam$-calculus used as a base calculus is a rather standard one featuring a choice operator $\{\; \}$. By using this operator, we can construct terms of the form $\{p_1t_1 , \ldots , p_nt_n \}$  that non-deterministically reduce to one of the terms $t_1, \ldots , t_n$. The particular result of the reduction is determined according to the probabilities $p_1, \ldots , p_n$, namely, $\{p_1t_1 , \ldots , p_nt_n \}$ will reduce to $t_i$ with probability  $p_i$.
We further define in \lamt{} an operator $\test_n$ to conduct an experiment comprising $n$ executions of a probabilistic program $s$ and to collect all their outputs in a tuple $\lan s_1 , \ldots , s_n\ran$. Such a tuple can be seen as an encoding of the frequency distribution of the types of possible outputs of $s$ relative to the specific executions of our original term $s$. We then define a $\trust$ operator to check whether the obtained frequency distribution is sufficiently close to a fixed desired probability distribution $\calp$. Technically, a term $\trust (\test _n   s) \calpe  :  \bool \, \calpe  $ reduces to $\trust \lan s_1 , \ldots , s_n\ran\calpe :  \bool \,\calpe$ and this latter term will reduce to the result of the check. In particular, if the distance between the obtained frequency distribution $\lan s_1 , \ldots , s_n\ran$ and the desired probability distribution $\calp$ is not larger that the threshold $\thr$, then $\trust \lan s_1 , \ldots , s_n\ran\calpe  :  \bool \,\calpe$  reduces to  the value $\tert$, otherwise to the value $\terf$.

The resulting notion of trust, implemented  by constructs of the form $\trust (\test _n s)\calpe  : \bool \,\calpe$, is parametric in two senses. First, these constructs do not simply encode a notion of ``trust in the program $s$'' but a notion of ``trust that the program $s$ computes its outputs as determined by the probability distribution $\calp$''. Second, by varying the number $n$ of tests conducted for any given experiment, also the degree of confidence with which one trusts the program $s$ varies. If $n$ is too small, any result will be considered untrustworthy, even if the program $s$ is defined to compute according to $\calp$: after a single throw of a coin, one would hardly trust the coin to be fair, even if the coin were indeed fair. If $n$ grows but not quite enough, results may lead  to trust the program $s$, but such trust would be weak inasmuch as it is based on samples of little significance: after few throws of the coin, a normal distribution of heads and tails may suggest that the coin is fair, but one should remain alert that it might not. The more $n$ grows, the more trust is acquired, since one gets closer to a good approximation of the actual probabilistic behaviour of the program $s$. Obviously, if $s$ is not defined to compute according to the probability distribution $\calp$, increasing $n$ will only confirm that we should not trust $s$ to correctly compute according to $\calp$. In order to formally frame these subtleties, we introduce a notion of {\it confidence} associated to trust. Intuitively, this notion enables us to reason about our confidence in the potential behaviour of the program and about the way it evolves when the thoroughness of our experiments on that program increases. Technically, it defines a partial order on programs and is formally based on the probability of obtaining $\tert$ when we apply $\trust $ to experiments comprising an increasing number of executions of a given program. By an application of the strong law of large numbers, we prove that {\it confidence} tends, for series of experiments consisting of an increasing number of executions, to attribute complete trust to those terms that are suitably defined with respect to the target probability distribution.

The rest of the article is structured as follows. In Section \ref{sec:calculus}, we present the calculus \lamt{} by formally specifying  how its terms (Subsection \ref{sec:terms}) and its typing system (Subsection \ref{sec:types}) are defined. In Section \ref{sec:computation}, we present the reduction rules that determine the computational behaviour of  \lamt{}-terms. In Section \ref{sec:properties}, we prove some essential properties of \lamt{} such as subject reduction and termination. In Section \ref{sec:strategy-and-examples},  we define the call-by-name strategy that will be used to have \lamt{}-terms compute according to their intended behaviour, we show that this strategy is terminating, and we discuss some examples (Subsection \ref{sec:computational-examples}). In Section \ref{sec:potential-analysis}, we define the notion of {\it confidence} in order to reason globally on the behaviour of terms on the basis of the runtime results obtained by employing the $\trust$ operator. In  Section \ref{sec:conditional-probability}, we show how conditional probabilities and the probabilities of conjunctions and disjunctions can be handled. In Section \ref{sec:related}, we present a brief survey of the main existing formal systems for probabilistic reasoning, probabilistic computation and trust, and we compare them with \lamt{}. Finally, we conclude in Section \ref{sec:conclusion} with some remarks concerning future work directions.

\section{The calculus \lamt}\label{sec:calculus}

We define now the calculus \lamt. First, we define its terms and explain their intended computational behaviour. Then, we introduce their types and briefly explain the kind of computational behaviour they are meant to describe. Finally, we present and explain the rules to assign types to the terms of the calculus.  

\subsection{Terms of  \lamt}\label{sec:terms}
The terms of \lamt{} comprise the traditional elements of $\lam$-calculus: variables $x,y,z\ldots$, lambda abstractions $\lam x.t$ and applications $tu$. We also include terms  $\tert$ and $\terf$ of a Boolean type.
The basic probabilistic component of \lamt is the constructor $\{\;\}$ which allows to build probabilistic terms of the form $\{p_1 t_1, \ldots , p_n t_n\}$. The intended behaviour of such a term is to non-deterministically reduce to one of the terms $t_1, \ldots , t_n$. The reduction of $\{p_1 t_1, \ldots , p_n t_n\}$ yields a particular term $t_i$ with $1\leq i\leq n $ non-deterministically selected. The probability that a particular term $t_i$ is selected is $p_i$. 
The calculus also includes an experiment operator $\test_n $ that allows to test the probabilistic behaviour of terms by executing its argument $n$ times: $\test _n t$ reduces to the tuple $\lan t , \ldots , t \ran$ containing $n$ copies of the term $t$, thus enabling $n$ different executions of $t$. The projection operator $\pi_j$ just  extracts the $j$th element of a tuple $\lan t_1 , \ldots , t_n \ran$. The result $\lan t_1 , \ldots , t_n \ran$ of an experiment on a probabilistic term can be further analysed by the operator $\trust$. When the argument of this operator is a tuple $\lan t_1 , \ldots , t_n \ran$ of outputs resulting from an experiment, $\trust$ executes a check comparing the frequency distribution of the outputs in the tuple with a fixed probability distribution. If the frequency distribution of the outputs agrees, modulo a threshold $\thr$, with the fixed probability distribution, the term $\trust\lan t_1 , \ldots , t_n \ran\calpe$ reduces to $\tert$, otherwise to $\terf$. As will be formally defined by the type assignment rules in Table \ref{tab:typing}, the $\trust$ operator will only be applied to tuples or terms that are supposed to reduce to a tuple -- such as experiments, suitable applications, projections, or probabilistic terms. All elements of a tuple, moreover, will be required to have the same type. This is due to the fact that a tuple is supposed here to encode the result of an experiment consisting of several executions of the {\it same} program.

We now formally define the grammar of terms. The grammar determines what are the legitimate forms that terms can have. As per usual practice, the typing rules in Table \ref{tab:typing} will further restrict the ways we can construct a term depending on the type of its immediate subterms.

\begin{definition}[Terms]\label{def:terms} The
terms of \lamt{} are defined by the following grammar:
\begin{align*}x,y,z\ldots  \;\; ::=  \;\; & x_0 \;\; \mid\;\; x_1 \;\; \mid\;\; x_2\;\; \ldots  
\\
t,s,u,v\ldots  \;\; ::= \;\; & \tert \;\; \mid\;\;  \terf \;\; \mid\;\;  x  \;\; \mid\;\; \lambda x . t \;\; \mid\;\; tu \;\; \mid \\ & \{p_1 t_1, \ldots , p_n t_n\}\;\; \mid\;\; \test_n t \;\; \mid\;\; \lan t_1 , \ldots , t_n\ran \;\; \mid \;\; t\pi_j\;\; \mid
\\ &
\trust t 
\end{align*}
where $n$ and $j$ are natural numbers such that $n\geq 1 $ and $j\geq 1$, and $p_1 , \ldots , p_n $ are rational  numbers such that $p_1 + \ldots + p_n = 1$.\end{definition}
For the sake of economy, we use a mixed notation in the grammars of terms and types. First, non-terminal symbols to the left of $::=$ also indicate the metavariables that we will use throughout the article. Second, for clarity, we slightly deviate from the traditional notation of generative grammars and distinguish different occurrences of the same non-terminal symbol by different subscripts, as in $\lan t_1, \ldots , t_n\ran$. Notice, moreover, that the indication of the probability distribution $\calp$ and of the threshold $\thr$ does not occur in untyped terms of the form $\trust t$. This is due to the fact that, technically, we treat $\calpe$ as a type annotation, just like  superscript types for variables in Church-style typing. The information  provided by $\calpe$ is indeed arbitrarily chosen at the moment of the typing of a term of the form $\trust t$ and essentially refers to the types of the terms involved.

\begin{definition}[Free and bound variables, closed terms]\label{def:free-bound}
The set of free variables $\FV (t)$ of a \lamt-term $t$ are inductively defined as follows:
\begin{itemize}
\item $\FV(\tert)=\FV(\terf)=\emptyset$
\item $\FV(x)=\{x\}$ 
\item $\FV(\lam x. t)=\FV(t)\setminus\{x\} $
\item $\FV(tu)=\FV(t)\cup \FV(u)$
\item $\FV(\{p_1 t_1, \ldots , p_n t_n\})=\FV(\lan t_1, \ldots , t_n\ran) = \FV(t_1)\cup  \ldots \cup \FV( t_n)$ 
\item $\FV(\test _n t )=\FV(t\pi_j) =\FV(\trust t)=\FV(t)$
\end{itemize}All the variables occurring in a term which are not free variables are the bound variables of the term. If a term has no free variables, we say it is closed.
\end{definition}

\begin{definition}[Substitution]
\label{def:substitution}
For any variable $x$ and terms $t,s$, the result of the substitution $t[s/x]$ of $s$ for $x$ inside $t$ is defined by induction on the structure of $t$ as follows.
\begin{itemize}
\item $\tert[s/x]=\tert$ and $\terf[s/x]=\terf$
\item $x[s/x]=s$ and $y[s/x]=y$ for any $y\neq x$
\item $(\lam x. t)[s/x]=\lam x. t$ and $(\lam y. t)[s/x]=\lam y. t[s/x]$ for any $y\neq x$
\item $(tu)[s/x]= t[s/x] u[s/x]$
\item $\{p_1 t_1, \ldots , p_n t_n\}[s/x]=\{p_1 t_1[s/x], \ldots , p_n t_n[s/x]\}$
\item $\lan t_1, \ldots , t_n\ran [s/x]=\lan t_1[s/x], \ldots , t_n[s/x]\ran $ 
\item $(\test _n t )[s/x]= \test _n t[s/x]$
\item $(t\pi_j)[s/x] =t[s/x]\pi_j$
\item $(\trust t)[s/x]=\trust t[s/x]$
\end{itemize}
\end{definition}

\subsection{Types of \lamt}\label{sec:types}

The types used to classify and operate on \lamt-terms comprise, as usual, arbitrarily many atomic types $q_1, q_2,q_3\ldots $, and the arrow types $A\impl B$ for functions that accept inputs of type $A$ and produce outputs of type $B$.
The type of probabilistic terms is the sum type $\cho_{i=1}^n A_i$ for functions that non-deterministically behave as specified by one of the types $A_i$. The type $(A)^n$ is used for tuples of length $n$ only containing elements of type $A$. This type is  intended for experiment results, that is, tuples of possibly different values obtained from a single term. Finally, the types of \lamt{} also include a Boolean type $\bool $ enriched by the indication of a probability distribution $p_1B_1, \ldots , p_m B_m$ over a list of types $B_1, \ldots ,  B_m$. This enriched version of the traditional Boolean type is used for the results of trustworthiness checks on probabilistic programs.  In simple words, a term of type $\bool (p_1B_1, \ldots , p_m B_m)(\thr)$ encodes an answer to the question ``do we trust this program to compute according to the distribution $p_1B_1, \ldots , p_m B_m$, modulo the error threshold $\thr$?''. 

\begin{definition}[Types]\label{def:types}
The types of  \lamt{} are defined by the following grammar:
\begin{align*}  
q\ldots  \quad ::=   \quad  &   q_0 \quad \mid\quad q_1 \quad \mid\quad q_2\quad \ldots  \\
A, B, C\ldots \quad ::=  \quad & q 
\quad \mid \quad A\impl B \quad \mid\quad \cho _{i=1}^{n} A_i  \quad \mid
\\&  (A)^n \quad \mid\quad \bool (p_1B_1, \ldots , p_n B_n) (\thr) 
\end{align*}where $n$ is a natural number, $n\geq 1$, $p_1,  \ldots , p_n , \thr \in [0,1]$ and $p_1 +  \ldots +p_n =1$.
\end{definition} 

We now introduce the notion of {\it subtype} in order to add some flexibility to the types in the calculus and, in particular, to be able to consider any possible outcome $t_j:A_j$ of a non-deterministic choice as a suitable reductum of a probabilistic choice term $\{p_1 t_1 , \ldots, p_n t_n\}:\cho _{i=1}^{n}A_i$. Usually, term reduction is type preserving. Thus, all functions and arguments compatible with the redex are compatible with the type of the reductum. In our case, though, we cannot enforce this simple identity since the type of $\{p_1 t_1 , \ldots , p_n t_n\}:\cho _{i=1}^{n}A_i$ is not the same as that of $t_j:A_j$, its possible reducti. Instead, we define and use a relation that specifies when terms of a certain type can be thought of and used as terms of another, more general type. For instance, given a term $t$ encoding the number $2$, it can be seen both as a term of type {\it integer} and as a term of type {\it natural}. In the second case, we are just being more specific about the nature of $t$. In any case, $t$ is a valid argument for a function accepting integers and for a function accepting natural numbers. It is reasonable then to define the typing rules in such a way that a function that accepts inputs of a certain type $T$ also accepts inputs of a more specific type than $T$ or, in other words, inputs with  a subtype of $T$ as type. The general idea is that, if $A$ is a subtype of $B$, then a term of type $A$ can be used in all cases in which a term of type $B$ is required. We intuitively justify the following definition of the subtyping relation accordingly. 
As for the boolean type, the information about $\calpe$ in a type $\bool \,\calpe$ does not change what constructs we can use on a term of type 
$\bool \,\calpe$. 
As for function types, in order to use a function of type $A_1\impl A_2 $ as if it were of type $B_1\impl B_2$ we need $A_1$ to be more general than $B_1$ -- which means that the function of type $A_1\impl A_2$ accepts any input accepted by a function of type $B_1\impl B_2$; and we need $A_2$ to be more specific than $B_2$ -- thus any output of our function of type $A_1\impl A_2 $ can be used as if it were an output of a function of type $B_1\impl B_2$. As for the choice types, we would like any possible outcome of a choice of type $\cho _{i=1}^{m} A_j$ to be usable as if it were an outcome of a choice of type $\cho _{i=1}^{n} B_i$. Hence we need to have that any $A_j$ can be thought of as some $B_i$ -- that is, for any $A_j$, it should be a subtype of some $B_i$. Thus we are guaranteed that any function or argument compatible with any outcome of the choice of type $\cho _{i=1}^{n} B_i$ will also be compatible with $A_j$.\footnote{Notice that, as discussed in detail in Section \ref{sec:related}, the subtype relation defined in \cite{dip20} is rather different from the one defined here since the former relation admits the possibility of eliminating choices from a choice term if they are not compatible with the argument given to the choice term itself.} Finally, as for the tuple types, a subtype of $(B)^n$ indicates that the elements of the tuple are of a subtype of $B$ and that the tuple has at least $n$ elements. Thus, if the $j${th} element can be extracted from a tuple of type $(B)^n$  and used as a term of type $B$, the same holds for a tuple typed by a subtype of $(B)^n$.   

\begin{definition}[Subtype]\label{def:subtype}
Given a relation $\cals$ defined on atomic types, the subtyping relation $<:$  extending $\cals$ to complex types is  defined as follows:
\begin{itemize}
\item for any type $A$, $A<:A$
\item for any three types $A,B,C$, if $A<:B$ and $B<:C$, then $A<:C$
\item if $A$ and $B$ are atomic types and $A\cals B$, then $A<:B$ 
\item if $B_1<:A_1$ and $A_2<:B_2$, then $A_1\impl A_2<:B_1\impl B_2$
\item if for any $A_j$ with $1\leq j\leq m$ there is a $B_i$ with $1\leq i\leq n$ such that $A_j<:B_i$, then $\cho _{i=1}^{m} A_j  <: \cho _{i=1}^{n} B_i$
\item if $A <: B$ and $m\geq n$, then $(A)^m <: (B)^n$
\end{itemize}
\end{definition}

The type assignment rules are presented in Table \ref{tab:typing}.

\begin{table}[ht]
  \centering

  \hrule
\medskip

\[\infer{x:A\Rightarrow x:A}{} \qquad  \infer{\Rightarrow \tert :\bool \,\calpe}{} \qquad \infer{\Rightarrow \terf :\bool \,\calpe}{} \]\[ 
\infer{\Gamma\Rightarrow \lam x^A . t : A\impl B}{\Gamma,x:A\Rightarrow t:B}
\qquad \infer{\Gamma,\Delta \Rightarrow tu : B}{\Gamma\Rightarrow t :A\impl B & \Delta \Rightarrow u: A}
\]
\[  
\infer{\Gamma_1,\dots,\Gamma_n\Rightarrow  \{p_1 t_1, \ldots , p_n t_n\}: \cho _{i=1}^{n}  A_i}{\Gamma_1\Rightarrow  t_1:A_1 & \ldots & \Gamma_n\Rightarrow t_n:A_n}\]
\[ 
\infer{\Gamma,\Delta\Rightarrow \{p_1 t_1, \ldots , p_n t_n\} t : \cho _{i=1}^{n}   D_i}{ \Gamma\Rightarrow \{p_1 t_1, \ldots , p_n t_n\}: \cho _{i=1}^{n}  (C_i\impl D_i) & \Delta\Rightarrow t:C} \]
\[
\infer{\Gamma\Rightarrow \test _n t : (A)^n}{\Gamma\Rightarrow t: A}\qquad 
\infer{\Gamma_1,\dots,\Gamma_n\Rightarrow \lan s_1 , \ldots , s_n \ran : (E)^n}{\Gamma_1\Rightarrow s_1:E_1  & \ldots &  \Gamma_n\Rightarrow s_n:E_n}\qquad 
\infer{\Gamma\Rightarrow t\pi_j:A}{\Gamma\Rightarrow t: (A)^n}\]
\[ \infer{\Gamma\Rightarrow \trust t\calpe :\bool \,\calpe}{\Gamma\Rightarrow t:(A)^n}
\]\begin{flushleft}where $p_1, \ldots , p_n\in[0,1]$ and $p_1+\ldots + p_n=1$; 
$A_1, \ldots , A_n$ are all subtypes of a type $A$; $C$ is a subtype of all types $C_1, \ldots , C_n$; $E_1, \ldots , E_n$ are subtypes of $E$; $1\leq j\leq n$; $\calp =  q_1B_1, \ldots , q_m B_m$ and $\thr$ is a rational number.
\end{flushleft}

\medskip

\hrule
\caption{Type assignment rules of \lamt}
\label{tab:typing}
\end{table}

The 0-premises rules in Table \ref{tab:typing} simply present the basic type assignments for atomic terms: traditional variables can be assigned any type, and $\tert,\terf$ can be assigned the type $\bool \,\calpe$ for any annotation $\calpe$. As usual, we type $\lam$-abstractions by arrow types. 

Probabilistic terms $\{p_1 t_1, \ldots , p_n t_n\}$ are typed by using the non-deterministic sum connective $\cho$, as expected. Applications of probabilistic choices over terms with arrow types $C_i\impl D_i$ are typed by the non-deterministic sum $\cho$ of the types $D_i$ of their outputs. Notice that the limit case in which $n=1$ gives the usual application rule for $\lam$-abstractions. 

The typing rule for $\test _n$ indicates that the result of the application of this operator to any term $t:A$ must have type $(A)^n$: this type is indeed the one used for tuples of length $n$ of objects of type $A$, which is exactly the kind of term that the reduction of $\test _n t$ will yield. In turn, the typing rule for $\lan s_1, \ldots , s_n \ran $ enables us to type the result of an experiment $\test _n s $. In this rule we require that the variables of $s_1, \ldots , s_n$ are bound since the test  $\test _n s $ can only be executed if all variables of $s$ are bound -- as it will be specified in Table \ref{tab:reductions}. Moreover, the type $(S)^n$ of the whole tuple $\lan s_1, \ldots , s_n \ran $ is based on a supertype $S$ of the types $S_1, \ldots , S_n$ since all elements of the tuple $s_1, \ldots , s_n $ are supposed to be obtained by reducing a unique term $s$ that we can assume to have type $S$.\footnote{In \lamt{}, a term of the form $\lan s_1, \ldots , s_n \ran $ is supposed to be produced by reducing a term $\test_n s$ that implements an experiment on a term $s:S$. As a consequence, even if the elements of the tuple may have different types $S_1, \ldots , S_n$, all their types are supposed to be subtypes of the type $S$ of $s$. Indeed, $ s_1:S_1, \ldots , s_n:S_n$ are supposed to be reduced forms of the term $s:S$. The tuple should then be typed as follows: $\lan s_1, \ldots , s_n \ran :(S)^n$. In some cases, though, there might also exist a type $S'\neq S$ such that  $S_1, \ldots , S_n$ are subtypes of $S'$. In these cases, one might also choose to type the tuple as follows: $\lan s_1, \ldots , s_n \ran :(S')^n$. Nevertheless, this possibility is not at all problematic with respect to the interaction between $\test $ and $\lan\;\ran$, as shown in Theorem \ref{thm:subject-reduction}.} The typing rule for the generalised projection $\pi_j$ just matches the usual typing rules for projections: the type of $t\pi_j$ is the type -- or a supertype of the type -- of each element of the tuple $t$.

The typing rule of the $\trust $ operator accomplishes two tasks at once. First, it guarantees that  $\trust $ is applied to a term $t$ of type $(A)^n$ -- which means that $t$ encodes an experiment consisting in testing $n$ times a term of type $A$. Second, it allows to freely choose the probability distribution and threshold indication $\calpe$ contained in the type $\bool \,\calpe$ to be assigned to $\trust t$. The chosen annotation $\calpe$, in particular, will indicate the probability distribution against which the frequency distribution encoded by the experiment $t$ will be tested for trustworthiness and the threshold $\thr$ employed during the test. 
Notice that since the information provided by $\calpe$ cannot be inferred from the structure of the terms of type $\bool  \,\calpe$, we annotate when required terms of this form by $\calpe$- -- as it is done for the types of variables in Church-style typing.

\section{Computational rules}\label{sec:computation}

We define now the dynamic computational behaviour of the terms of \lamt{} by introducing the reduction rules for typed terms, and the reduction strategy that we will employ to evaluate them. Before this, we present some definitions  that will be employed in the reduction of applications of the $\trust$ operator.

We introduce the notion of {\it value}, that is, a term in the correct form to be an output, i.e. the final result of a computation.
\begin{definition}[Values]\label{def:value}
Any \lamt-term $t:T$ is a value if, and only if, it has one of the following forms:
\begin{itemize}
\item $\tert$ or $\terf$
\item $\lam x . t$
\item $\lan s_1, \ldots , s_n \ran $ where each $s_1, \ldots , s_n$ is a value
\end{itemize}
\end{definition}

When we want to reduce a term of the form $ \trust \lan s_1 , \ldots , s_n \ran  \calpe  $, we need to consider the experiment result $\lan s_1 , \ldots , s_n \ran$ and decide whether it establishes or not that the term on which the experiment was run is trustworthy. In order to do so, we compute a probabilistic divergence between the frequency of the types of the terms $ s_1 , \ldots , s_n$ and the target distribution $\calp$. Afterwards, we compare the value of the divergence with the threshold $\thr$. The divergence that we use here is the {\it  total variation distance} ($\dist$). This choice is simply due to the fact that $\dist$ is a quite simple divergence and a very commonly used one, but other divergences can be used instead of $\dist$ without introducing any further complications in the calculus. Here are two formal definitions determining the method we use to compute $\dist$ between a list of terms and a target distribution.

\begin{definition}[Total variation distance]
The total variation distance $\dist (\calp\parallel \calp')$ between any two probability distributions $\calp = (p_1 T_1, \ldots , p_n T_n)$ and $\calp '= (p'_1 T'_1, \ldots , p'_{n'} T'_{n'})$ is defined as follows: 
\begin{itemize}
\item if neither $\{T_1, \ldots ,  T_n\}\subseteq \{T'_1, \ldots , T'_{n'}\} $ nor  $\{T'_1, \ldots , T'_{n'}\} \subseteq\{T_1, \ldots ,  T_n\}$, then $\dist (\calp\parallel \calp')=1$
\item otherwise, let 
\begin{itemize}
\item $S_1, \ldots , S_m $ be an enumeration of the types occurring in $\{T_1, \ldots ,  T_n\}\cup \{T'_1, \ldots , T'_{n'}\} $
\item $q_i$ (respectively $q'_i$), for $1\leq i\leq m$, be the sum of the probabilities attributed to $S_i$ in $\calp$ (respectively in $\calp '$)  if any,  $0$ otherwise
\end{itemize}and then 
\[\dist (\calp\parallel \calp')= \sup  _{i=1}^{m} (\lvert {q_{i}-q'_{i}}\rvert)\]
where $\vert p\rvert=p$ if $p$ is positive and $\vert p\rvert=-p$ if $p$ is negative.
\end{itemize}
\end{definition}In this definition we employ the enumeration $S_1, \ldots , S_m $ and the sums of probabilities $q_i$ and $q'_i$ rather than directly employing $T_1, \ldots , T_n, T'_1, \ldots , T'_{n'}$ and the probabilities $p_i$ and $p'_i$ because a single type $S_i$ might occur more than once in each distribution $\calp$ and $\calp'$ -- this is not forbidden, indeed, by Definition \ref{def:types} -- and because types might be differently ordered in  $\calp$ and $\calp'$.

\begin{definition}[Frequency--probability correspondence]
For any term of the form $t=\lan s_1 , \ldots , s_n \ran $ and any distribution $\calp = q_1B_1, \ldots , q_m B_m$, $t$ encodes a frequency that corresponds to $\calp $ (in short, $t \encodes\calp$) if, and only if, 
\begin{itemize}
\item all terms $s_1, \ldots , s_n$ are values
\item for $1\leq i\leq n$, $S_i$ is the type of  $s_i$
\item for any $j$ such that $1\leq j\leq m$, the occurrences of subtypes of $B_j$ in $S_1, \ldots , S_n$ are exactly $q_j\cdot n$    
\end{itemize}
\end{definition}

The reduction rules are presented in Table \ref{tab:reductions}.

\begin{table}[ht]
  \centering

  \hrule
\medskip

$(\lam x.t )u \;\;\red_1\;\; t[u/x]$  
\medskip\medskip





$\{p_1 t_1, \ldots , p_n t_n\}  \;\;\red_{p_i}\;\; t_i  $ where $1\leq i \leq n $\medskip\medskip


$\test _n s\;\;\red_1\;\;\lan s , \ldots , s \ran$

\medskip\medskip

$\lan s_1 , \ldots , s_n \ran\pi _i\;\;\red_1\;\; s_i \;\;$ where $1\leq i\leq n$\medskip\medskip

$\trust \lan s_1 , \ldots , s_n \ran  \calpe  \;\;\red_1\;\; \tert $ 

if $\lan s_1 , \ldots , s_n \ran\encodes \calp'$ and $\dist( \calp \parallel \calp ') \leq \thr $
\medskip\medskip

$\trust \lan s_1 , \ldots , s_n \ran   \calpe \;\;\red_1\;\; \terf $ 

if $\lan s_1 , \ldots , s_n \ran \encodes \calp' $ and $\dist( \calp \parallel \calp ')  > \thr $\medskip\medskip


\medskip

\hrule
\caption{Reduction rules of \lamt}
\label{tab:reductions}
\end{table} 

\begin{definition}[Redexes, reducti and reductions]
\label{def:redex}
We call a term of the form displayed to the left of any occurrence of $\red$ in Table \ref{tab:reductions} a redex and the corresponding term displayed to the right of $\red$, its reductum. If $t$ and $t'$ are terms such that $t\red_p t'$ is an instance of a reduction rule in Table \ref{tab:reductions}, $t$ occurs inside a term $s$ (possibly, $t=s$) and $s'$ is obtained by replacing one occurrence of $t$ in $s$ by $t'$, then we say that $s$ reduces to $s'$ with probability $p$. We indicate this by writing $s\red_p s'$.    
\end{definition}

The reduction for applications of $\lam$-abstractions is the usual one. Simply notice that we specify that $(\lam x . t)u$ reduces to $t[u/x]$ with a probability of $1$ since this is not a probabilistic reduction but a deterministic one. We do the same for all other deterministic reductions. 
The reduction of a probabilistic choice, on the other hand, is  governed by a probabilistic reduction rule: the notation $\{p_1 t_1, \ldots , p_n t_n\} \;\red_{p_i}\; t_i  $ means in particular that the term $\{p_1 t_1, \ldots , p_n t_n\}$ can non-deterministically reduce to any term among $t_1, \ldots , t_n$ and that the probability that it will actually reduce to a particular term $t_i$ is exactly $p_i$. 

The reduction rule for $\test _n$ enables us to reduce a term $\test_n t$ to a tuple containing $n$ times the term $t$. 


The reduction rule for the projection $\pi_j$ simply enables us to extract the $j$th element of the tuple to which the projection is applied.

The reductions for $\trust$ are triggered when the operator is applied to a tuple only containing values $s_1, \ldots , s_n$ and hence codifying a frequency distribution over the possible outputs of the term from which we obtained $s_1, \ldots , s_n$.\footnote{Notice that, according to the typing rule for $\lan\; \ran$, if the term $\lan s_1, \ldots , s_n \ran$ is typable, then there is a type $S$ such that the types of $ s_1, \ldots , s_n$ are subtypes of $S$. Thus we know, according to the typing rule for $\{\;\}$, that there are typable probabilistic choice terms of the form $\{p_1s_1, \ldots , p_ns_n\}$. Since $s_1, \ldots , s_n$ are values that can be obtained by reducing these probabilistic choice terms, we know that if $\lan s_1, \ldots , s_n \ran$ is typable, then there is at least one term such that the tuple $\lan s_1, \ldots , s_n \ran$ encodes one possible outcome of the experiment on this term.} If this distribution actually matches -- modulo the threshold $\thr$ -- the desired probability distribution $\calp$ against which we are conducting the trustworthiness check, then the result of the application of $\trust$ is $\tert$; otherwise, the result is $\terf$.


\section{Properties of the calculus \lamt}
\label{sec:properties}

The main results of this section are the Thm. \ref{thm:subject-reduction} for subject reduction and Thm. \ref{thm:logical-termination} for termination. The former guarantees that typing is robust with respect to the computation, the latter that any typed \lamt-term can be reduced in finite time to a term that does not contain redexes. Before proving these results we show that substitutions preserve typability modulo the subtype relation.


\begin{lemma}[Substitution lemma]\label{lem:substitution}
For any $x:A$, $t:B$, $s:A'$ such that $A'<:A$, there exists $B'<:B$ such that $t[s/x]:B'$.
\end{lemma}
\begin{proof}We prove the following stronger statement: for any $x:A$, $t:B$, $s:A'$ such that $A'<:A$, there exists $B'<:B$ such that $t[s/x]:B'$. Moreover, if $B=B_1\impl B_2$ then $B'=B_1\impl B_2'$, and if $B=\cho _{i=1}^{n}   C_i$, then $B'=\cho _{i=1}^{n}   C'_i$.

The proof is by induction on the structure of $t$.
\begin{itemize}
\item $t=y$. Either $x\neq y$ or $x=y$. If $x\neq y$, then $t[s/x]=y[s/x]=y:B$ and the statement holds due to the reflexivity of $<:$. If, on the other hand, $x=y$, then $t[s/x]=x[s/x]=s:A'$. But since $t=y=x$, we also have that  $A=B$, which together with $A'<:A$ gives us that $A'<:B$. Hence, the statement holds.

\item $t=\tert$ or $t=\terf$. In both these cases, $t[s/x]=t$ and the statement holds because of the reflexivity of $<:$.  

\item $t=\lam y.u $. Either $x\neq y$ or $x=y$. If $x=y$, then $t[s/x]=t$ and the statement holds due to the reflexivity of $<:$. If, on the other hand, $x\neq y$, then $t[s/x]=(\lam y.u)[s/x]=\lam y.u[s/x]$. Now, since $t=\lam y.u :B$, $B$ must be of the form $B_1\impl B_2$, $y:B_1$ and $u:B_2$. By inductive hypothesis, the statement holds for $u[s/x]$ and thus $u[s/x]:B_2'$ such that $B_2'<:B_2$. But then, $\lam y u[s/x]:B_1\impl B_2'$. Since, by Definition \ref{def:subtype}, $B_1\impl B_2'<:B_1\impl B_2$ when $B_2'<:B_2$, we have that $t[s/x]=\lam y u[s/x]:B_1\impl B_2'<:B_1\impl B_2=B$, and thus the statement holds.

\item $t= uv$. Then  $t[s/x]=u[s/x]v[s/x]$.  By inductive hypothesis, the statement holds for $u[s/x]$ and $v[s/x]$. Since $t=uv :B$, $B=\cho _{i=1}^{n}   D_i$, $u:\cho _{i=1}^{n}  (C_i\impl D_i) $ and $v:C$ where $C$ is a subtype of each $C_i$ for $1\leq i \leq n$. Therefore, $u[s/x]:\cho _{i=1}^{n}  (C_i\impl D'_i) $ and $v[s/x]:C'$ where $\cho _{i=1}^{n}  (C_i\impl D'_i) <:\cho _{i=1}^{n}  (C_i\impl D_i)$ and $C'<:C$. We have then that $C'$ is still a subtype of each $C_i$ for $1\leq i \leq n$. But then, $u[s/x]v[s/x]: \cho _{i=1}^{n}   D'_i$ where, by Definition \ref{def:subtype}, $D'_i<:D_i$ because otherwise we would not have that $\cho _{i=1}^{n}  (C_i\impl D'_i) <:\cho _{i=1}^{n}  (C_i\impl D_i)$. By putting everything together, we obtain that $t[s/x]=u[s/x]v[s/x]: \cho _{i=1}^{n}   D'_i<: \cho _{i=1}^{n} D_i =B$ and hence that the statement holds.

\item $t= \{p_1 t_1 , \ldots , p_nt_n\}$. Then  $t[s/x]=\{p_1 t_1[s/x] , \ldots , p_nt_n[s/x]\}$.  By inductive hypothesis, the statement holds for all $t_i[s/x]$. Since $t= \{p_1 t_1 , \ldots , p_nt_n\} :B$, $B=\cho _{i=1}^{n}     C_i$ and each $t_i:  C_i$. Hence we have $t_i[s/x]:  C_i'$ where $   C'_i <:   C_i$. But then, $t[s/x]=\{p_1 t_1[s/x] , \ldots , p_nt_n[s/x]\}: \cho _{i=1}^{n}   D'_i<: \cho _{i=1}^{n} D_i =B$ and hence the statement holds.

\item $t= \test _n u $. Then  $t[s/x]=\test _n u[s/x]$.  By inductive hypothesis, the statement holds for $u[s/x]$. Since $t= \test _n u:B$, $B=(C)^n$ and $u:C$. Hence we have $u[s/x]:C'$ where $ C' <: C$. But then, $t[s/x]=\test _n u[s/x]: (C')^n<:(C)^n =B$ and hence the statement holds.

\item    $t= \lan  t_1 , \ldots , t_n \ran  $. Since $t= \lan  t_1 , \ldots , t_n \ran  :B$, $B=(C)^n$ such that, for each $t_i:C'$, $C'<:C$ holds. By inductive hypothesis, the statement holds for each $t_i[s/x]:C'$ such that $1\leq i\leq n$. Therefore, $t_i[s/x]:C''$ where $ C'' <: C'$. But then, since $C''<:C'<:C$, $t[s/x]= \lan  t_1[s/x] , \ldots , t_n[s/x] \ran: (C)^n<:(C)^n =B$ is well-typed and the statement holds.


\item $t= u\pi_j$. Then  $t[s/x]= u[s/x]\pi _j$.  By inductive hypothesis, the statement holds for $u[s/x]$. Since $t=  u\pi_j : B$, we have $u:(B)^m$. Hence we have $u[s/x]:C$ where $ C <: (B)^m$. But, by Definition \ref{def:subtype}, $C$ must be $(B')_l$ where $B'<:B$ and $l\geq m$. Hence, $u[s/x]\pi_j$ is well defined -- since $j\leq m\leq l$ -- and $u[s/x]\pi_j:B'$. Putting everything together we have that $t[s/x]=u[s/x]\pi_j:B'<:B$ and so the statement holds.

\item $t= \trust u$. Then  $t[s/x]=\trust u[s/x]$.  By inductive hypothesis, the statement holds for $u[s/x]$. Since $t= \trust u:B$, we have $B=\bool \,\calpe$ and $u:(C)^m$. Hence we have $u[s/x]:C'$ where $ C' <: (C)^m$. But, by Definition \ref{def:subtype}, $C$ must be $(B')_l$ where $B'<:B$ and $l\geq m$. Hence, $\trust u[s/x]:\bool \,\calpe$ is still well defined and, by the reflexivity of $<:$, the statement holds.

\end{itemize}
\end{proof}

\begin{theorem}[Subject reduction]\label{thm:subject-reduction}
For any term $v:T$, if $v\red_{p}v'$, then there is a type $T'$ such that $T'<:T$ and $v':T'$.
\end{theorem}
\begin{proof}We reason by cases on the reduction rule $v\red_{p}v'$.
\begin{itemize}
\item $(\lam x.t )u \red_1   t[u/x]$. Suppose that $x:A$ and $t:B$. The type of $\lam x.t $ is then $A\impl B$ and the type of  $(\lam x.t )u$ is $ B$. Since, by the typing rule for applications, we must have that the type of $u$ is a subtype of the type $A$ of $x$, by Lemma \ref{lem:substitution}, $t[u/x]:B'$ where $B'<:B$.  Hence the statement holds.

\item $\{p_1 t_1, \ldots , p_n t_n\}    \red_{p_i}   t_i  $ where $1\leq i \leq n $. Then  $v: \cho _{i=1}^{n}A_i$ and each $t_i$ must be of type $A_i$. Since, by Definition \ref{def:subtype}, $A_j<:\cho _{i=1}^{n}A_i$ holds for any $1\leq j\leq n $, also the statement holds.


\item $\test _n s  \red_1  \lan s , \ldots , s \ran $. Then  $v: (A)^n$ and $v':(A)^n$, and, by the reflexivity of $<:$, the statement holds.

\item $\lan s_1 , \ldots , s_n \ran\pi _i  \red_1   s_i $. Then $v: A$ and $v':A$, and, by the reflexivity of $<:$, the statement holds.

\item $\trust \lan s_1 , \ldots ,s_n \ran \calpe   \red_1   \tert $ or $\trust \lan s_1 , \ldots , s_n \ran \calpe   \red_1   \terf $. Then  $v: \bool \,\calpe$ and $v': \bool \,\calpe$, and, by Definition \ref{def:subtype}, the statement holds.

\end{itemize}\end{proof}

%
%
%
%
%
%
%
%
%

%

Finally, we show that each typed term of \lamt{} can be reduced, in a finite number of reduction steps, to a term that does not contain redexes. Let us first introduce the notion of logical complexity of a redex, which will be essential in proving the result.

\begin{definition}[Logical complexity of a redex]
The logical complexity of a redex is defined according to the relative reduction rule:  
\begin{itemize}
\item $(\lam x.t )u $ where $\lam x.t:A\impl B$. Then the logical complexity is the number of symbol occurrences in  $A\impl B$.

\item $\{p_1 t_1, \ldots , p_n t_n\} $ where $\{p_1 t_1, \ldots , p_n t_n\} :\cho _{i=1}^n A_i$. Then the logical complexity is the number of symbol occurrences in  $\cho _{i=1}^n A_i$. 

\item $\test _n s $ where $\test _n s : (A)^n$. Then the logical complexity is the number of symbol occurrences in $(A)^n$.

\item $\lan s_1 , \ldots , s_n \ran\pi _i $ where $\lan s_1 , \ldots , s_n \ran : (A)^n$. Then the logical complexity is the number of symbol occurrences in $(A)^n$.

\item $\trust \lan s_1 , \ldots ,s_n \ran \calpe $. The logical complexity of any redex due to an application of the $\trust $ operator is $0$.
\end{itemize}\end{definition}


\begin{theorem}[Termination of the logical reduction strategy]\label{thm:logical-termination}
For any typed \lamt-term $t:T$, there is a term $t'$ such that $t$ reduces to $t'$ in a finite number of reduction steps and $t'$ does not contain any redexes.
\end{theorem} 
\begin{proof}
The complexity of a term is given by the triple $(\mu, \nu , \tau)$ where $\mu$ is the maximal logical complexity of the redexes in the term, $\nu$ is the number of redexes in the term with logical complexity $\mu$, and $\tau$ is the number of occurrences of $\test$ redexes with maximal logical complexity in the term. We order these triples lexicographically. We show that if we always reduce the rightmost redex of maximal logical complexity in the term, we decrease the following complexity measure at each reduction step. From this it follows that the reduction terminates. We reason by cases on the form of the redex we reduce.  

\begin{itemize}
\item $(\lam x.s)u \red_1   s[u/x]$. No maximal redex is duplicated since $(\lam x.s)u$ is rightmost among the redexes with maximal logical complexity. The complexity of the other new redexes that can be generated by this reduction is either determined by the type of $s[u/x]$ or by the type of $u$. In either case, the logical complexity is lower than that of the redex just reduced. Therefore, we either reduced the maximal logical complexity  $\mu$ of the redexes in the term or the  number $\nu$ of redexes  in the term with maximal logical complexity.

\item $\{p_1 t_1, \ldots , p_n t_n\}    \red_{p_i}   t_i  $. No redex is duplicated by this reduction.  The only  new redex that can be generated by this reduction has complexity determined by the type of $t_i$. The  logical complexity of these is lower than that of the redex just reduced. Therefore, we either reduced the maximal logical complexity  $\mu$ of the redexes in the term or the number $\nu$ of redexes  in the term with maximal logical complexity.

\item $\test _n s  \red_1  \lan s , \ldots , s \ran $. No maximal redex is duplicated since $\test _n s$ is rightmost among the redexes with maximal logical complexity. The only new redex that can be generated by this reduction is the projection of $\lan s , \ldots , s \ran$. If the reduction does not produce this new maximal projection redex, we reduced a maximal redex without producing or duplicating any maximal redex. In this case, we reduced the complexity of the term as usual. If, on the other hand, the reduction produces a new maximal redex -- that is, the projection of $\lan s , \ldots , s \ran$ -- then we reduced a maximal redex, produced a new redex with the same complexity, but also reduced the number of $\test $ redexes with maximal logical complexity in the term. Hence, the complexity of the term decreases in this case as well.

\item $\lan s_1 , \ldots , s_n \ran\pi _i  \red_1   s_i $. No redex is duplicated by this reduction.  The only  new redex that can be generated by this reduction has complexity determined by the type of $s_i$. The  logical complexity of this new redex is lower than that of the redex just reduced. Therefore, we either reduced the maximal logical complexity  $\mu$ of the redexes in the term or the number $\nu$ of redexes  in the term with maximal logical complexity.

\item $\trust \lan s_1 , \ldots ,s_n \ran \calpe   \red_1   \tert $ and  $\trust \lan s_1 , \ldots , s_n \ran \calpe   \red_1   \terf $. No redex is duplicated or generated by this reduction. Therefore, we either reduced the maximal logical complexity  $\mu$ of the redexes in the term or the number $\nu$ of redexes  in the term with maximal logical complexity.
\end{itemize}
\end{proof}

\section{Reduction strategy and examples}

\label{sec:strategy-and-examples}

For the reduction of \lamt-terms we employ a call-by-name-style strategy adapted to our language. This is an obvious choice if we consider that our main constraint is to give a meaningful behaviour to the experiment construct $\test _n$. Indeed, we need to avoid the case in which the argument of $\test _n$ is reduced to a value before the experiment is conducted. If we did not enforce this condition, we would admit series of reductions during which $\test _n t$ is reduced to $\test _n t'$, where $t'$ is a value obtained by reducing $t$, and only afterwards  $\test _n t'$ is reduced to  $\lan t', \ldots , t'\ran $. But a series of reductions of this kind corresponds to a computation during which we execute the program $t$ once and then consider the frequency distribution corresponding to $n$ identical copies of its output $t'$. What we would like to have, instead, is a computation in which we copy $t$ itself $n$ times and then evaluate autonomously the $n$ copies by $n$ different  evaluations, possibly obtaining different results in case $t$ is a probabilistic term. Consider, for instance, the very simple but meaningful example of an experiment consisting of ten throws of a fair coin -- technically, ten executions of a probabilistic term encoding a fair coin. The unwanted kind of computation is one in which we throw the coin once, obtain heads, for instance, and consider the frequency distribution given by assuming that the coin yielded heads ten times during the experiment. This is clearly not desirable. The desired kind of computation is, instead, one in which we throw the coin ten times and consider the frequency distribution given by the ten results obtained by ten different throws. 

In order to force computations involving $\test _n$ to match this second case -- and thus avoid the trivialisation of the computational behaviour of $\test_n $ -- we define a strategy that forces the redex $\test _n t$ to be reduced before $t$ is reduced to a value. This already has a call-by-name flavour since the argument of  $\test _n $ is not reduced before we reduce the application of the operator to the argument -- which would correspond to a call-by-value strategy for $\test _n$. We adopt a call-by-name strategy also for other operators to avoid the undesirable cases in which the arguments of $\test _n$ are reduced to a value just because some other construct reduced them before passing them to $\test _n$. For an example of this second problem, consider the term $(\lam x. \test_n x) t$. If we reduce the argument $t$ to a value before reducing the application of $(\lam x. \test _n x)$, we obtain $\test _n t'$ where $t'$ is a value. When we reduce $\test _n t'$ to $\lan t' , \ldots , t' \ran $ we obtain a value encoding the frequency distribution according to which $t$ gave $t'$ as output $n$ times. But this is the case only because we reduced $t$ before actually running the experiment. If we did not do that, $t$, which could be a probabilistic term, might have reduced in several different ways and the experiment might have yielded a non-trivial frequency distribution. By adopting a call-by-name strategy also for applications of $\lam$-abstraction and probabilistic choices, we avoid also this undesired behaviour.

\begin{definition}[Call-by-name  (CBN) reduction strategy]\label{def:cbn}
The reduction relation $\red_p$ is defined by extending the 0-premises reduction rules in Table \ref{tab:reductions} by the following reduction rules:
\begin{itemize}
\item $\vcenter{\infer{ts\red_p t's}{t\red_p t'}}$ 
\item $\vcenter{\infer{\lan \ldots , t, \ldots \ran \red_p \lan \ldots , t',  \ldots \ran}{t\red_p t'}}$ 

where $t$ is the leftmost element of $\lan \ldots , t, \ldots \ran$ which is not a value\footnote{Notice that, for the sake of simplicity, we also adopt here the convention that, inside tuples, we always reduce the leftmost available redex. Since the reductions of different elements of a tuple do not interact in any way, by doing this we simply eliminate irrelevant instances of non-determinism in the reduction of terms.}

\item $\vcenter{\infer{\trust \,t\,\calpe \;\red_p\; \trust\, t'\,\calpe}{t\red_p t'}}$

\item $\vcenter{\infer{t\pi_j \red_p  t'\pi_j}{t\red_p t'}}$

where $t$ is not of the form $ \lan t_1 , \ldots , t_m \ran$

\end{itemize}
\end{definition}

Finally, we prove that the computation of a typed closed term according to our call-by-name strategy never stops without yielding a value.
\begin{theorem}[Progress]\label{thm:progress}
For any closed typed \lamt-term $t:T$, either  there exists a $t'$ such that $t\red_p t'$ according to Definition \ref{def:cbn}, or $t$ is a value.
\end{theorem}
\begin{proof}The proof is by induction on the structure of $t$. 

\begin{itemize}
\item $t=y$. This is impossible since $t$ is closed.

\item $t=\tert$ or $t=\terf$. In both these cases, $t$ is a value and the statement holds.  

\item $t=\lam y.u $. Then $t$ is a value and the statement holds.

\item $t= uv$. Since $uv$ is closed, so is $u$. Hence, by inductive hypothesis, either $u$ is a value or there exists a $u'$ for which $u\red_p u'$. If $u$ is a value, it must be of the form $\lam x. u_1$ because otherwise could not be applied to $v$. Hence $uv =  (\lam x. u_1)v\red_1 u_1[v/x]$ and the statement holds. If $u\red_p u'$, by Definition \ref{def:cbn}, we have that $uv\red_p u'v$ and the statement holds.

\item $t= \{p_1 t_1 , \ldots , p_nt_n\}$. Then $t\red_{p_i} t_i$ and the statement holds.

\item $t= \test _n u $. Since $t $ is closed, also $u$ must be closed, and thus $t=\test _n u\red_1 \lan u, \ldots , u\ran $. Therefore, the statement holds.

\item $t= \lan  t_1 , \ldots , t_n \ran  $. Since $t$ is closed, each $t_i$ for $1\leq i\leq n$ is a closed term. Hence, by inductive hypothesis, either each $t_i$ is a value, or there is a leftmost $t_i$ such that there exists $t'_i$ for which $t_i\red _pt'_i$. If each $t_i$ is a value, then $t= \lan  t_1 , \ldots , t_n \ran $ is a value and the statement holds. If there is a leftmost $t_i$ such that $t_i\red _pt'_i$, then, by Definition \ref{def:cbn}, we have that $t= \lan  t_1 , \ldots ,t_i , \ldots , t_n \ran  \red_p  \lan  t_1 , \ldots ,t'_i , \ldots , t_n \ran  $ and the statement holds.

\item $t= u\pi_j$. Since $u\pi_j$ is closed, so is $u$. Hence, by inductive hypothesis, either $u$ is a value, or there exists a $u'$ for which $u\red_p u'$. If $u$ is a value, it must be of the form $\lan  u_1 , \ldots , u_m \ran$ because otherwise $\pi_j$ could not be applied to it. Hence $t=\lan  u_1 , \ldots , u_m \ran\pi_j\red _1 u_n$ and the statement holds. If, on the other hand, there exists a $u'$ for which $u\red_p u'$, then we have two more cases. Either $u$ is  of the form $\lan  u_1 , \ldots , u_m \ran$ or it is not. If it is, we have $t=\lan  u_1 , \ldots , u_m \ran\pi_j\red _1 u_n$ and the statement holds. If it is not, by Definition \ref{def:cbn}, we have $t= u\pi_j \red_p \trust u'\pi_j$ and the statement holds.

\item $t= \trust u\calpe$.  Since $ \trust  u\calpe$ is closed, so is $u$. Hence, by inductive hypothesis, either $u$ is a value, or there exists a $u'$ for which $u\red_p u'$. If $u$ is a value, it must be of the form $\lan  u_1 , \ldots , u_n \ran$ because otherwise $\trust$ could not be applied to it. Hence either $\lan  u_1 , \ldots , u_n \ran\encodes \calp'$ and $\mid \calp - \calp'\mid  \leq \thr $, which implies that  $\trust u\calpe\red _1 \tert$, or  $\lan  u_1 , \ldots , u_n \ran \encodes \calp' $ and $\mid \calp - \calp'\mid  > \thr $, which implies that $\trust u \calpe\red _1 \terf$. In either case, the statement holds. If, on the other hand, there exists a $u'$ for which $u\red_p u'$, then, by Definition \ref{def:cbn}, we have that $t=\trust u \red_p \trust u'$ and the statement holds.

\end{itemize}\end{proof}

We now prove that also the reduction of a term according to the  call-by-name  strategy terminates. The idea of the proof is to associate all reduction steps in the CBN reduction of $t$ to a reduction step in a logical reduction of $t$ in such a way that only a finite number of reduction steps of the CBN reduction are associated to each step of the logical reduction---in doing this, we also suppose that the two reductions agree on all probabilistic choices. The possibility of associating reductions in this way enables us to conclude that the number of reduction steps in the CBN reduction is infinite only if the number of reduction step in the logical reduction is infinite.

In order to prove the theorem, we begin by defining a procedure for associating redexes occurring in different terms $t_\text{CBN}'$ and $t_\text{Logic}'$ obtained by reducing the same term $t$ by the CBN, respectively, logical reduction strategy. The association is done by a procedure for labelling redexes and variables occurring in $t_\text{CBN}'$ and $t_\text{Logic}'$. This labelling will be used to track the evolution of a redex even if its body slightly changes due to other redexes being reduced. Intuitively, we label all redexes in the term $t$ by applying a label to the outermost symbol $\lam$, $\pi_j$, $\{$, $\test$, or $\trust $ of the redex. 
Only if a redex is reduced or erased by a reduction, the label disappears from the term. For instance, if $r$ is a redex, $x$ does not occur in $t$, and we execute the following reduction: $(\lam x. t)r \red_1 t$, then the label of $r$ and the label of the redex $(\lam x. t)r$ itself will disappear from the term. If a redex is copied by the reduction of another redex, we will have more redexes in the obtained term with the same label. For instance,  if  $r$ is a redex and we execute the following reduction: $(\lam x. \lan x,xt\ran)r \red_1 \lan r,rt\ran$, then both copies of $r$ in the obtained  term $\lan r,rt\ran$ will have the same label, which is the same label that $r$ had in the term $(\lam x. \lan x,xt\ran)r$. Finally, if new redexes are generated by a reduction -- as in $(\lam x. xt)(\lam y. s) \red_1 (\lam y. s)t$ -- we label them by new labels in such a way that we do not distinguish between redexes that can be considered as duplicates of the same redex occurrence. In order to do this we also need to proceed similarly with variable occurrences in order to track their duplications.

%
%
%

\begin{definition}[Labelled, twin reduction of a term]
\label{def:labelling}
  
Given a term $t$, we label all its redex occurrences by a distinct label $\alpha, \beta, \gamma\dots $ We then duplicate the resulting labelled term and name $t_\text{CBN}^0$ one copy and $t_\text{Logic}^0$  the other copy.

For any $n$, in order to obtain the term $t_\text{CBN}^{n+1}$, we reduce one redex in $t_\text{CBN}^n$ by following the CBN strategy presented in Definition \ref{def:cbn}. As for $t_\text{Logic}^{n+1}$ , suppose that the redex of $t_\text{CBN}^n$ that we have just reduced has label $\alpha$. We have two cases:
\begin{itemize}
\item If the label $\alpha $ occurs in $t_\text{Logic}^n$, in order to obtain the term $t_\text{Logic}^{n+1}$, we reduce possibly several redexes in $t_\text{Logic}^n$ by following the logical reduction strategy presented in Theorem \ref{thm:logical-termination} until we have reduced all the redexes labelled $\alpha$. 

\item If no redex with label $\alpha$ occurs in $t_\text{Logic}^n$, we do not reduce anything and take $t_\text{Logic}^{n+1}=t_\text{Logic}^n$.
\end{itemize}

The labelling of new redexes proceeds as follows.
\begin{itemize}
\item If the label $\alpha $ occurs in $t_\text{Logic}^n$, then in  $t_\text{CBN}^{n+1}$ we have that: 
\begin{itemize}
\item {\bf Body-generated redexes.} For any label $\alpha $, all redexes generated by the reduction of redexes $u\pi_j$, $\{\ldots \}$, or $\test u$ with label $\alpha $ are labelled by a unique new label $\alpha_0$. The label  $\alpha _0$ is also used for all redexes generated by the reducti $u[s/x]$ of all redexes $(\lam x . u )s$ labelled $\alpha $. 

\item {\bf Argument-generated redexes.} The redexes generated by the term $s$ which is being substituted by the substitution $[s/x]$ produced by redexes $(\lam x . u )s$ labelled $\alpha $ are labelled by new different labels $\alpha_1, \ldots , \alpha _n $ -- all different both from $\alpha$ and from $\alpha _0$ -- for each different label of an occurrence of $x$ in $u$. In other words, two redexes generated by the term $s$ being substituted for the variable $x$ will have the same label if, and only if, the two occurrences of $x$ that we replace by $s$ have the same label.
\end{itemize}

After having labelled $t_\text{CBN}^{n+1}$, we label the new redexes obtained in $t_\text{Logic}^{n+1}$ by copying the labelling used in $t_\text{CBN}^{n+1}$ as follows:
\begin{itemize}
\item {\bf Body-generated redexes.} If a body-generated redex in $t_\text{CBN}^{n+1}$ is labelled $\beta$, then also the corresponding body-generated redex in $t_\text{Logic}^{n+1}$ will be labelled $\beta$.

\item {\bf Argument-generated redexes.} 
If an argument-generated redex in $t_\text{CBN}^{i+1}$ is obtained by replacing a variable labelled $\gamma$ and is labelled $\beta$, then also any argument-generated redex in $t_\text{Logic}^{i+1}$ obtained  by replacing a variable labelled $\gamma$ will be labelled $\beta$.
\end{itemize}

\item If the label $\alpha $ does not occur in $t_\text{Logic}^n$ because it has already been reduced at an earlier stage $m<n$ of the twin, labelled reduction, then we label the new redexes generated in $t_\text{CBN}^{n+1}$ by copying the labelling employed in $t_\text{Logic}^m$ as follows:
\begin{itemize}
\item {\bf Body-generated redexes.} If a body-generated redex in $t_\text{Logic}^m$ is labelled $\beta$, then also the corresponding body-generated redex in $t_\text{CBN}^{n+1}$ will be labelled $\beta$.

\item {\bf Argument-generated redexes.} 
If an argument-generated redex in $t_\text{Logic}^m$ is obtained by replacing a variable labelled $\gamma$ and is labelled $\beta$, then also any argument-generated redex in $t_\text{CBN}^{n+1}$ obtained  by replacing a variable labelled $\gamma$ will be labelled $\beta$.
\end{itemize}

\item If the label $\alpha $ does not occur in $t_\text{Logic}^n$ but it has not been reduced at an earlier stage $m<n$, then we introduce completely new labels for the new redexes, but we still use the same label for body-generated redexes obtained by redexes with the same label, and for argument-generated redexes obtained by replacing variables with the same label.
\end{itemize}
\end{definition}

The main normalisation argument relies on three rather intuitive statements. We formally prove them before moving to the proof of the main theorem. 
\begin{lemma}\label{lem:finite-duplication}
No (finite or infinite) \emph{labelled, twin reduction} can duplicate a label an infinite number of times.
\end{lemma}
\begin{proof} Labels, according to Definition \ref{def:labelling}, are only duplicated when a reduction duplicate terms. Newly generated redexes in a reduction sequence $t^0\red_{p_0} \dots \red_{p_{n}}t^{n+1}$ are never assigned a label which was previously assigned to some other redex. Hence, if two redexes have the same label, then they have been generated together by the reduction of another redex. This implies that each label can only be duplicated only a finite amount of times even if our sequence of reductions is not finite.
\end{proof}

\begin{lemma}\label{lem:no-reduce-erase}
If during a \emph{labelled, twin reduction} we reduce a redex with label $\alpha $ in a term $t_\text{CBN}^n$, then it is impossible that all redexes with label $\alpha $ have been erased during the reduction of $t_\text{Logic}^0$ into $t_\text{Logic}^n$
\end{lemma}
\begin{proof} 
It is enough to show that if a redex $r$ is erased by the logical strategy, then also the CBN strategy erases it before reducing it. But this is true. Indeed, if the logical strategy erases a redex $r$, it is because it occurs inside another redex $r'$ of the form $(\lam x .u )v $ where $x$ does not occur in $u$, $\lan \ldots \ran \pi_j$, $\{\ldots \} $ or $\trust u$. But if this is the case, since the CBN strategy never reduces  terms inside a choice operator $\{\;\}$, arguments of applications, arguments of projections or arguments of the $\trust $ operator, we are guaranteed that the CBN strategy will erase the redex $r$ before reducing it.
\end{proof}

\begin{lemma}\label{lem:prev-reduce}
If during a \emph{labelled, twin reduction} no redex in $t_\text{Logic}^n$ has the label $\alpha$ of the redex in $t_\text{CBN}^n$ that we are reducing, then some redexes with label $\alpha $ occurred during the reduction of $t_\text{Logic}^0$ into $t_\text{Logic}^n$ and these redexes have been reduced -- as opposed to simply erased.
\end{lemma}
\begin{proof} 
Lemma \ref{lem:no-reduce-erase} guarantees that, if a redex is reduced by the CBN strategy, the logical strategy will not erase it but reduce it. Hence, in order to prove the statement, it is enough to show that, according to Definition \ref{def:labelling}, we never assign to a redex that will be reduced during the CBN reduction a label $\alpha $ that cannot be assigned to any redex occurring during the logical reduction. We prove that this holds by induction on the number $n$ of the stage of the computation:
\begin{itemize}
\item Base case. We first consider body-generated redexes. It is enough to show that if the first redex reduced during the CBN reduction generates a new redex, also the logical reduction of the redexes with the same label does. But this follows from Lemma \ref{lem:no-reduce-erase} and from the fact that, even though the logical reduction could reduce other redexes before the one which is being reduced in the CBN reduction, once the conditions under which a reductum of a redex is a redex are met, they cannot be made false by reducing other redexes. As for  argument-generated redexes, we must also consider the duplication of the argument of the reduction and the behaviour of the labelled variables  that are replaced by the argument during the reduction. The only new potential problem that arises in this case is that a variable label that occurs in $t_\text{CBN}^0$ is erased from $t_\text{Logic}^0$ before the last redex that yields $t_\text{Logic}^1$ is reduced. But this implies that the new redex occurring in $t_\text{CBN}^1$ will be erased later on in the CBN reduction. And since, as we argued in proving $(2)$, the CBN strategy never reduces a redex occurring in a term that can be erased, the statement holds.

\item Inductive step. Suppose that in all reduction stages preceding the stage $i+1$ each redex label introduced by the CBN strategy has also been assigned to some redex in the logical reduction. Then the same argument used for the base case applies to the stage $n$. Indeed, by inductive hypothesis, a term $t_\text{Logic}^m$, for some $m\leq n$, contains a redex $r'$ labelled just like the redex $r$ that we are reducing in $t_\text{CBN}^n$ to obtain $t_\text{CBN}^{n+1}$. The reduction of $r'$, as we argued in the base case, will eventually produce a redex that we can label by the same label used for $r$. Also for the argument-generated redexes, the considerations presented for the base case hold for the inductive step as well. Indeed, if the logical reduction erases all variables with the same label as the variables that we replace while reducing the redex that yields $t_\text{CBN}^{n+1}$ form $t_\text{CBN}^n$, then, by $(2)$, the argument-generated redex occurring in $t_\text{CBN}^{n+1}$ will be erased later on during the CBN reduction and never reduced. But since, as we argued in proving Lemma \ref{lem:no-reduce-erase}, the CBN strategy never reduces a redex occurring in a term that can be erased, the statement holds also in this case.
\end{itemize}
\end{proof}

We can finally prove the main normalisation theorem for our CBN strategy.

\begin{theorem}\label{thm:cbn-termination}For any typed \lamt-term $t:T$, the call-by-name strategy in Definition \ref{def:cbn} applied to $t$ never produces an infinite sequence of reductions.
\end{theorem}
\begin{proof} In order to prove that the call-by-name (CBN) reduction defined in Definition \ref{def:cbn} of any term $t$ terminates, we use an indirect argument. We show in particular that, if there is an infinite CBN reduction of $t:T$, then there is also an infinite reduction of $t$ according to the logical reduction strategy. But since this would contradict Theorem \ref{thm:logical-termination}, we conclude that there is no infinite CBN reduction of $t$, for any term $t:T$.

We reason indirectly. Suppose indeed that the CBN strategy from Definition \ref{def:cbn} applied to the a term $t:T$ produces an infinite sequence of reductions $\sigma_\text{CBN}$. Consider then the labelled, twin reduction $\sigma_\text{Logic}$ of $t:T$ that follows the logical strategy introduced in Theorem \ref{thm:logical-termination} and in which all reductions of a probabilistic choice operator match the corresponding ones occurring during $\sigma_\text{CBN}$ -- that is, we consider the sequence $\sigma_\text{Logic}$ during which a redex $\{p_1t_1 , \ldots p_n t_n \}$ with label $\alpha$ reduces to $t_i$  if, and only if, the corresponding redex  $\{p_1t'_1 , \ldots p_n t'_n \}$ with label $\alpha$ reduces to $t'_i$ during $\sigma_\text{CBN}$. We can always suppose that the choice operators in $\sigma_\text{Logic}$ behave just as in $\sigma_\text{CBN}$ since, for our argument, it is enough that there exists at least one sequence of reductions $\sigma_\text{Logic}$ with the suitable properties. Now, by Lemmata \ref{lem:no-reduce-erase} and \ref{lem:prev-reduce}, each redex reduced in $\sigma_\text{CBN}$ is associated to at least one redex reduced in $\sigma_\text{Logic}$. Moreover, by Lemma \ref{lem:finite-duplication}, only a finite number of redexes reduced during $\sigma_\text{CBN}$ is associated to the same redex reduced during $\sigma_\text{Logic}$. Therefore, since $\sigma_\text{CBN}$ is supposed to be infinite, also  $\sigma_\text{Logic}$ must be infinite. But this contradicts the fact that, according to  Theorem \ref{thm:logical-termination}, $\sigma_\text{Logic}$ must be finite. In conclusion then, the CBN strategy from Definition \ref{def:cbn} applied to a term $t:T$ cannot produce an infinite sequence of reductions.

\end{proof}

\subsection{Examples of computation}\label{sec:computational-examples}

\begin{example}
Consider our language extended with the types $\mathbb{1},\mathbb{2}, \mathbb{3}, \mathbb{4}, \mathbb{5}, \mathbb{6}$ and the terms \[1:\mathbb{1}\qquad 2:\mathbb{2}\qquad 3:\mathbb{3}\qquad 4:\mathbb{4}\qquad 5:\mathbb{5}\qquad 6:\mathbb{6}\]

Construct a probabilistic term \[\{\frac{1}{6} 1  , \frac{1}{6} 2  ,\frac{1}{6}3  , \frac{1}{6}4  ,\frac{1}{6}5  , \frac{1}{6}6  \}: \mathbb{1}\oplus\mathbb{2}\oplus \mathbb{3}\oplus \mathbb{4}\oplus \mathbb{5}\oplus \mathbb{6}\] implementing a fair dice. The experiment consisting in throwing  the dice four times is then implemented by the following term:
\[\test _4\{\frac{1}{6} 1  , \frac{1}{6} 2  ,\frac{1}{6}3  , \frac{1}{6}4  ,\frac{1}{6}5  , \frac{1}{6}6  \}:(\mathbb{1}\oplus\mathbb{2}\oplus \mathbb{3}\oplus \mathbb{4}\oplus \mathbb{5}\oplus \mathbb{6})^4\]  
According to our CBN strategy, the term reduces to the following one:
{\tiny \[\lan \{\frac{1}{6} 1  , \frac{1}{6} 2  ,\frac{1}{6}3  , \frac{1}{6}4  ,\frac{1}{6}5  , \frac{1}{6}6  \},  \{\frac{1}{6} 1  , \frac{1}{6} 2  ,\frac{1}{6}3  , \frac{1}{6}4  ,\frac{1}{6}5  , \frac{1}{6}6  \}, \{\frac{1}{6} 1  , \frac{1}{6} 2  ,\frac{1}{6}3  , \frac{1}{6}4  ,\frac{1}{6}5  , \frac{1}{6}6  \}, \{\frac{1}{6} 1  , \frac{1}{6} 2  ,\frac{1}{6}3  , \frac{1}{6}4  ,\frac{1}{6}5  , \frac{1}{6}6  \}\ran\]}still of type $(\mathbb{1}\oplus\mathbb{2}\oplus \mathbb{3}\oplus \mathbb{4}\oplus \mathbb{5}\oplus \mathbb{6})^4$ and then, depending on what probabilistic reductions actually get triggered, to a term of the following form with probability $\frac{1}{1296}$:
\[\lan 2, 5, 6, 3\ran:(\mathbb{1}\oplus\mathbb{2}\oplus \mathbb{3}\oplus \mathbb{4}\oplus \mathbb{5}\oplus \mathbb{6})^4\]where the numbers inside the tuple $\lan \;\; \ran$ actually resulting from the reduction may vary.

In the calculus, though, we also have the operator $\trust$ that enables us to actually check whether the behaviour of a term during an experiment leads us to trust the term to compute certain values with a certain probability. Indeed, we can construct the term
\[\trust \test _4\{\frac{1}{6} 1  , \frac{1}{6} 2  ,\frac{1}{6}3  , \frac{1}{6}4  ,\frac{1}{6}5  , \frac{1}{6}6  \}:\bool(\frac{1}{6}\mathbb{1}, \frac{1}{6}\mathbb{2}, \frac{1}{6} \mathbb{3}, \frac{1}{6} \mathbb{4}, \frac{1}{6} \mathbb{5}, \frac{1}{6} \mathbb{6})(0)\]that checks whether the experiment consisting of four executions of the term $\{\frac{1}{6} 1  , \frac{1}{6} 2  ,\frac{1}{6}3  , \frac{1}{6}4  ,\frac{1}{6}5  , \frac{1}{6}6  \}$ actually gives an acceptable tuple of results in order for us to trust the term to actually compute according to the distribution $\frac{1}{6}\mathbb{1}, \frac{1}{6}\mathbb{2}, \frac{1}{6} \mathbb{3}, \frac{1}{6} \mathbb{4}, \frac{1}{6} \mathbb{5}, \frac{1}{6} \mathbb{6}$.

Supposing, for instance, that the four executions go exactly as before, we obtain:
\[\trust \lan 2,5,6,3\ran:\bool(\frac{1}{6}\mathbb{1}, \frac{1}{6}\mathbb{2}, \frac{1}{6} \mathbb{3}, \frac{1}{6} \mathbb{4}, \frac{1}{6} \mathbb{5}, \frac{1}{6} \mathbb{6})( 0)\]
The $\trust $ operator now checks whether the frequency distribution encoded by $ \lan 2,5,6,3\ran$ matches the desired probability distribution  $\frac{1}{6}\mathbb{1}, \frac{1}{6}\mathbb{2}, \frac{1}{6} \mathbb{3}, \frac{1}{6} \mathbb{4}, \frac{1}{6} \mathbb{5}, \frac{1}{6} \mathbb{6}$ specified in the type. Since the test only features four executions of the term -- that is, four throws of the dice -- the check will result negative in this case, and the whole term will reduce to $\terf :\bool(\frac{1}{6}\mathbb{1}, \frac{1}{6}\mathbb{2}, \frac{1}{6} \mathbb{3}, \frac{1}{6} \mathbb{4}, \frac{1}{6} \mathbb{5}, \frac{1}{6} \mathbb{6})(0)$ meaning that at the moment, after only four executions, we cannot trust the term to actually compute according to  the specified probability distribution.
\end{example}

\begin{example}
Suppose one wants to check whether the results of the dice are evenly distributed between even results and odd results. First, one needs to introduce two new atomic types: $\mathbb{even}$ and $ \mathbb{odd}$. Then, one can  specify, by the relation $\cals$ for atomic subtypes at the base of the subtyping relation $<:$, what results of the dice are even and what are odd:
\[\mathbb{1}\,\cals\,\mathbb{odd} \qquad 
\mathbb{2}\,\cals\,\mathbb{even}\qquad 
 \mathbb{3}\,\cals\,\mathbb{odd}\qquad 
 \mathbb{4}\,\cals\,\mathbb{even}\qquad 
 \mathbb{5}\,\cals\,\mathbb{odd}\qquad 
 \mathbb{6}\,\cals\,\mathbb{even}\] 
in such a way that the following subtyping relations hold:
\[\mathbb{1}\,<:\,\mathbb{odd} \qquad 
\mathbb{2}\,<:\,\mathbb{even}\qquad 
 \mathbb{3}\,<:\,\mathbb{odd}\qquad 
 \mathbb{4}\,<:\,\mathbb{even}\qquad 
 \mathbb{5}\,<:\,\mathbb{odd}\qquad 
 \mathbb{6}\,<:\,\mathbb{even}\] 
 
Now one can test the results of the dice with respect to the desired probability of even results and odd results. We do it by  constructing the following term:

\[\trust \test _{10} \{\frac{1}{6} 1  , \frac{1}{6} 2  ,\frac{1}{6}3  , \frac{1}{6}4  ,\frac{1}{6}5  , \frac{1}{6}6  \} : \bool(\frac{1}{2}\mathbb{even}, \frac{1}{2}\mathbb{odd})(0)\]

To be a bit more careful,  execute the term ten times during the experiment in order to have a slightly more meaningful insight on its probabilistic behaviour. After several reduction steps, the term might reduce, for instance to

\[\trust \lan 2, 1, 5, 4, 6, 3,4 ,4, 1, 5\ran : \bool(\frac{1}{2}\mathbb{even}, \frac{1}{2}\mathbb{odd})(0)\]

which gives us a frequency distribution where $\frac{5}{10}$ of the results are even and $\frac{5}{10}$ are odd, as formally specified by the types of the terms in the tuple modulo our subtyping relation. This clearly agrees with the desired probability distribution  $\frac{1}{2}\mathbb{even}, \frac{1}{2}\mathbb{odd}$ and hence triggers the reduction of the previous term into the term $\tert : \bool(\frac{1}{2}\mathbb{even}, \frac{1}{2}\mathbb{odd})(0)$ that specifies that we trust our original term to yield results that are even half of the times and odd the other half of the times.
\end{example}

\section{Reasoning on the potential behaviour  of terms}\label{sec:potential-analysis}

The runtime analysis that can be conducted inside the calculus by the operators $\test$ and $\trust$ is clearly limited to specific executions of probabilistic programs. Even though we can use the  operator $\test_n$ with greater and greater numbers $n$ in order to study better approximations of the expected probabilistic behaviour of a term, we will never have a complete formal guarantee that a probabilistic term will behave exactly as expected during a single computation or that a particular runtime experiment will yield exactly the desired outcome. Probabilistic terms, indeed, behave in a probabilistic way. Nevertheless, the calculus \lamt{} gives us the possibility to reason about the behaviour of terms by an analysis that evades this problem by considering the potential probabilistic behaviour of a term in general and not its actual behaviour during one computation. This can be done by studying the tree of the possible reductions of a term $t$, see \cite{dhw05,bdgm16} for semantics based on this idea. Such a tree has the term $t$ as its root, the nodes of the tree are the terms to which the original term can reduce to, and each edge represents the one-step reduction that applied to a node yields one of the children. The leaves of the tree are then the values that can be obtained by reducing the term $t$ at the root. Each edge is moreover labelled by the probability that the associated reduction will be actually triggered. The probability to obtain one particular leaf $l$ by reducing the term $t$ at the root is computed by multiplying all labels that we encounter in the path between the root $t$ and the leaf $l$. Since a single value $s$, obtainable by reducing  $t$, might occur at several leafs, the probability to obtain the particular value $s$ by reducing $t$ is computed by summing up the probabilities of all the leafs at which we have $s$.

This kind of analysis formalises the idea that by conducting a series of experiments comprising an increasing number of tests each, the possible results of the experiments better and better approximate the expected probability distribution computed by the term under analysis. We will show in particular that, under reasonable conditions, a term deemed trustworthy at a stage $n$ of the analysis will be deemed trustworthy also at all later stages.

We begin by formally defining the tree displaying all possible reductions of a term.

\begin{definition}[Reduction tree]\label{def:reduction-tree}
A reduction tree is a tree with \lamt-terms as nodes and each edge labelled by a rational number. For any typed \lamt-term $t:T$ with no free variables, the reduction tree of $t$ is denoted $\calt _t$ and is inductively defined as follows:
\begin{itemize}
\item the root of $\calt _t$ is $t$
\item for each node $u$ in $\calt _t$ and for each instance  $u \red _p u'$ of a reduction rule, there is a child $u'$ of $u$ such that the edge between $u$ and $u'$ is labelled by $p$.
\end{itemize}

\end{definition}

Notice that if a reduction tree contains a node $u$, its subtree rooted at $u$ is always exactly $\calt_u$. Moreover, if more than one reduction of a term $u$ yields the same term $u'$, we will have in $\calt_u$ a distinct children $u'$ of $u$ for each distinct  reduction. For instance, $\calt _{\{\frac{1}{2}\tert, \frac{1}{2}\tert\}}$ is
\[\begin{tikzpicture}[level distance=55pt,sibling distance=30pt] \Tree [.$\{\frac{1}{2}\tert,\frac{1}{2}\tert\}$ \edge node[left] {$\frac{1}{2}\;$}; [.$\tert$ ] \edge node[right] {$\;\frac{1}{2}$}; [.$\tert$ ] ]\end{tikzpicture}
\]

It is also easy to see that the sum of the labels of the edges from a node to its children is always $1$. Indeed, due to the adopted reduction stategy, to a \lamt-term we can either apply only one reduction with probability $1$, or several probabilistic reductions. In the second case, the sum of the probabilities of these reductions is still $1$. 


Let us consider some examples of reduction trees. 

\begin{example}
Suppose to add to the calculus the atomic types $H, T$ and two constant terms $\hcoin:H$, for heads, and $\tcoin:T$, for tails. The reduction tree of the term $ (\lam x . \{\frac{2}{3}\hcoin, \frac{1}{3}\tcoin\}) \tert$ -- which simply takes the input $\tert$ as a start signal, discards it, and then reduces to a biased coin -- is the following:
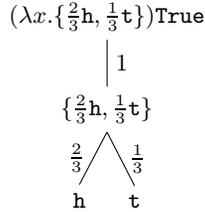
\[\begin{tikzpicture}[level distance=35pt,sibling distance=10pt] \Tree
[.$(\lam x.\{\frac{2}{3}\hcoin,\frac{1}{3}\tcoin\})\tert$ \edge node[right] {$1$}; [.$\{\frac{2}{3}\hcoin,\frac{1}{3}\tcoin\}$ \edge node[left] {$\frac{2}{3}$}; [.$\hcoin$ ] \edge node[right] {$\frac{1}{3}$}; [.$\tcoin$ ] ] ]\end{tikzpicture}\]
\end{example}

\begin{example}
The reduction tree of the slightly more complex term 

$$(\lam x.\{\frac{1}{8}\hcoin,\frac{3}{8}\tcoin,\frac{2}{8}\hcoin,\frac{2}{8}\{\frac{1}{2}(\lam y.y)\tcoin,\frac{1}{2}\hcoin\}\})\tert$$

is the following:
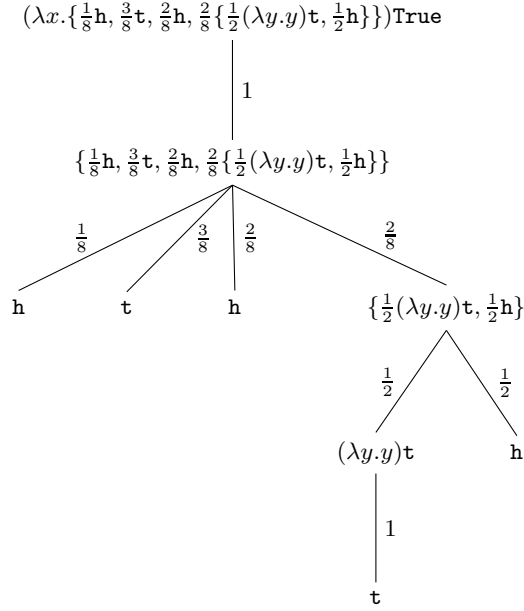
\[\begin{tikzpicture}[level distance=55pt,sibling distance=30pt] \Tree [.$(\lam x.\{\frac{1}{8}\hcoin,\frac{3}{8}\tcoin,\frac{2}{8}\hcoin,\frac{2}{8}\{\frac{1}{2}(\lam y.y)\tcoin,\frac{1}{2}\hcoin\}\})\tert$ \edge node[right] {$1$}; [.$\{\frac{1}{8}\hcoin,\frac{3}{8}\tcoin,\frac{2}{8}\hcoin,\frac{2}{8}\{\frac{1}{2}(\lam y.y)\tcoin,\frac{1}{2}\hcoin\}\}$ \edge node[left] {$\frac{1}{8}\;\;\;\;$}; [.$\hcoin$ ] \edge node[right] {$\;\frac{3}{8}$}; [.$\tcoin$ ] \edge node[right] {$\frac{2}{8}$}; [.$\hcoin$ ] \edge node[right] {$\;\;\;\;\;\frac{2}{8}$}; [.$\{\frac{1}{2}(\lam y.y)\tcoin,\frac{1}{2}\hcoin\}$ \edge node[left] {$\frac{1}{2}\;$}; [.$(\lam y.y)\tcoin$ \edge node[right] {$1$}; [.$\tcoin$ ] ] \edge node[right] {$\;\frac{1}{2}$}; [.$\hcoin$ ] ] ] ]\end{tikzpicture}\]
According to last tree, the probability of reducing the term according to, say, the rightmost path leading to the constant $h$ is $\frac{1}{2}\cdot\frac{2}{8}\cdot 1 = \frac{2}{16}= \frac{1}{8}$. The global probability of reducing the term to the constant $h$ -- that is, the probability of obtaining the output $h$ by any computation -- is $(\frac{1}{8}\cdot 1)+ (\frac{2}{8}\cdot 1)+(\frac{1}{2}\cdot\frac{2}{8}\cdot 1) = \frac{1}{2}$. On the other hand, the global probability of reducing the term to the constant $t$ is  $(\frac{3}{8}\cdot 1)+ (1\cdot \frac{1}{2}\cdot \frac{2}{8}\cdot  1) = \frac{1}{2}$.  
\end{example}

Let us now consider some examples that will naturally lead to the definition of a notion of confidence that induces an ordering between terms with respect to given probability distributions. This notion of confidence is based on the notion of trust formalised inside the calculus but enables us to abstract from particular computations and to consider  the global computational behaviour of a probabilistic term.


\begin{example}

Consider again the constants $\hcoin:H$, for heads, and $\tcoin:T$, for tails, and
the term $ \{\frac{1}{2}\hcoin, \frac{1}{2}\tcoin\}:H\oplus T$  implementing a fair coin. Consider then the term $\test _4 \{\frac{1}{2}\hcoin{}, \frac{1}{2}\tcoin\}:(H\oplus T)^4$ implementing an experiment consisting of four tosses of the fair coin. The reduction tree of $\test _4 \{\frac{1}{2}\hcoin{}, \frac{1}{2}\tcoin\}$ is of the form:{\tiny
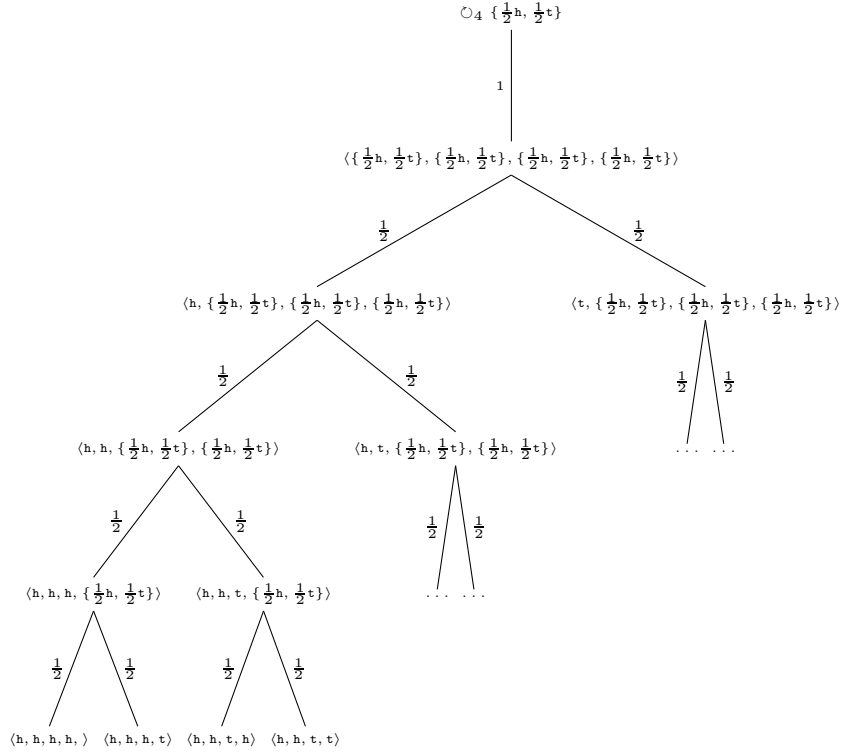
\[\begin{tikzpicture}[level distance=55pt,sibling distance=0pt] \Tree [.$\test _4\{\frac{1}{2}\hcoin{},\frac{1}{2}\tcoin\}$ \edge node[left] {$1$}; [.$\lan\{\frac{1}{2}\hcoin,\frac{1}{2}\tcoin\},\{\frac{1}{2}\hcoin,\frac{1}{2}\tcoin\},\{\frac{1}{2}\hcoin,\frac{1}{2}\tcoin\},\{\frac{1}{2}\hcoin,\frac{1}{2}\tcoin\}\ran$ \edge node[left] {$\frac{1}{2}\;\;\;$}; [.$\lan \hcoin,\{\frac{1}{2}\hcoin,\frac{1}{2}\tcoin\},\{\frac{1}{2}\hcoin,\frac{1}{2}\tcoin\},\{\frac{1}{2}\hcoin,\frac{1}{2}\tcoin\}\ran$ 
\edge node[left] {$\frac{1}{2}\;\;$}; [.$\lan \hcoin,\hcoin,\{\frac{1}{2}\hcoin,\frac{1}{2}\tcoin\},\{\frac{1}{2}\hcoin,\frac{1}{2}\tcoin\}\ran$ 
\edge node[left] {$\frac{1}{2}\;$}; [.$\lan \hcoin,\hcoin,\hcoin,\{\frac{1}{2}\hcoin,\frac{1}{2}\tcoin\}\ran$ \edge node[left] {$\frac{1}{2}$}; [.$\lan \hcoin,\hcoin,\hcoin,\hcoin,\ran$ ] \edge node[right] {$\frac{1}{2}$}; [.$\lan \hcoin,\hcoin,\hcoin,\tcoin\ran$ ] ] \edge node[right] {$\;\frac{1}{2}$}; [.$\lan \hcoin,\hcoin,\tcoin,\{\frac{1}{2}\hcoin,\frac{1}{2}\tcoin\}\ran$ \edge node[left] {$\frac{1}{2}$}; [.$\lan \hcoin,\hcoin,\tcoin,\hcoin\ran$ ] \edge node[right] {$\frac{1}{2}$}; [.$\lan \hcoin,\hcoin,\tcoin,\tcoin\ran$ ] ] 
] \edge node[right] {$\;\;\frac{1}{2}$}; [.$\lan \hcoin,\tcoin,\{\frac{1}{2}\hcoin,\frac{1}{2}\tcoin\},\{\frac{1}{2}\hcoin,\frac{1}{2}\tcoin\}\ran$ \edge node[left] {$\frac{1}{2}$}; [.$\ldots$ ] \edge node[right] {$\frac{1}{2}$}; [.$\ldots$ ] ]
] \edge node[right] {$\;\;\;\frac{1}{2}$}; [.$\lan \tcoin,\{\frac{1}{2}\hcoin,\frac{1}{2}\tcoin\},\{\frac{1}{2}\hcoin,\frac{1}{2}\tcoin\},\{\frac{1}{2}\hcoin,\frac{1}{2}\tcoin\}\ran$ \edge node[left] {$\frac{1}{2}$}; [.$\ldots$ ] \edge node[right] {$\frac{1}{2}$}; [.$\ldots$ ] ] ] ]\end{tikzpicture}
\]}
Here is a complete list of the leaves of this tree:

\begin{center}
\begin{tabular}{cccccc}
$ \lan \hcoin,\hcoin,\hcoin,\hcoin \ran$
&
$ \lan \hcoin,\hcoin,\hcoin,\tcoin \ran$
&
$ \lan \hcoin,\hcoin,\tcoin,\hcoin \ran$
&
$ \lan \hcoin,\hcoin,\tcoin,\tcoin \ran$
&
$ \lan \hcoin,\tcoin,\hcoin,\hcoin \ran$
&
$ \lan \hcoin,\tcoin,\hcoin,\tcoin \ran$
\\
$ \lan \hcoin,\tcoin,\tcoin,\hcoin \ran$
&
$ \lan \hcoin,\tcoin,\tcoin,\tcoin \ran$
&
$ \lan \tcoin,\hcoin,\hcoin,\hcoin \ran$
&
$ \lan \tcoin,\hcoin,\hcoin,\tcoin \ran$
&
$ \lan \tcoin,\hcoin,\tcoin,\hcoin \ran$
&
$ \lan \tcoin,\hcoin,\tcoin,\tcoin \ran$
\\
$ \lan \tcoin,\tcoin,\hcoin,\hcoin \ran$
&
$ \lan \tcoin,\tcoin,\hcoin,\tcoin \ran$
&
$ \lan \tcoin,\tcoin,\tcoin,\hcoin \ran$
&
$ \lan \tcoin,\tcoin,\tcoin,\tcoin \ran$&&
\end{tabular}
\end{center}

Each one of these values are obtained with a probability of $\frac{1}{16}$ since each choice between $\hcoin$ and $\tcoin$ happens with a probability of $\frac{1}{2}$. Consider then the term $\trust \test _4 \{\frac{1}{2}\hcoin, \frac{1}{2}\tcoin\}(\frac{1}{2}H,\frac{1}{2}T)(\frac{1}{4}) :\bool (\frac{1}{2}H,\frac{1}{2}T)(\frac{1}{4})$ that checks whether the result of the experiment $\test _4 \{\frac{1}{2}\hcoin, \frac{1}{2}\tcoin\}$ contains two outputs of type $H$ and two of type $T$ and admits one mistake over the four outputs. In other words,  $\trust \test _4 \{\frac{1}{2}\hcoin, \frac{1}{2}\tcoin\} (\frac{1}{2}H,\frac{1}{2}T)(\frac{1}{4}) $  will consider any tuple containing at most three outputs of the same type as a positive result of the experiment -- and hence reduce to $\tert :\bool (\frac{1}{2}H,\frac{1}{2}T)(\frac{1}{4})$. Therefore, computing the probability of obtaining the output $\tert :\bool (\frac{1}{2}H,\frac{1}{2}T)(\frac{1}{4})$ from  $ \trust \test _4 \{\frac{1}{2}h, \frac{1}{2}t\} (\frac{1}{2}H,\frac{1}{2}T)(\frac{1}{4}) $ is only a simple matter of adding up the probabilities of all tuples that encode a suitable frequency distribution. In our case, the suitable tuples are $14$. Therefore, the probability that we will trust the term $ \{\frac{1}{2}\hcoin, \frac{1}{2}\tcoin\}$ to compute according to the probability distribution $\frac{1}{2}H,\frac{1}{2}T$ modulo a mistake of $\frac{1}{4}$ is $\frac{14}{16} =0,875$. This is a good result, but still based on a few executions of the term: only $4$. Nevertheless, if we construct a similar term $ \trust \test _{8} \{\frac{1}{2}\hcoin, \frac{1}{2}\tcoin\} (\frac{1}{2}H,\frac{1}{2}T)(\frac{1}{4}) $ implementing an experiment consisting of $8$ executions of our term $\{\frac{1}{2}\hcoin, \frac{1}{2}\tcoin\}:H\oplus T$ we can see that the probability increases to $ \frac{238}{256}=0,9296875$. If we raise the number of executions to $12$ and consider the reduction tree of $ \trust \test _{12} \{\frac{1}{2}\hcoin, \frac{1}{2}\tcoin\} (\frac{1}{2}H,\frac{1}{2}T)(\frac{1}{4}) $, we get an even greater probability of obtaining a positive result: $\frac{3938}{4096}=0,96142578125$. And, in general, we have that the greater the number $n$, the greater the probability of obtaining $\tert$ as output of $ \trust \test _{n} \{\frac{1}{2}h, \frac{1}{2}t\} (\frac{1}{2}H,\frac{1}{2}T)(\frac{1}{4})$.

\end{example}

Once the key notions and calculations that will be employed in the definition of our confidence order have been provided, we present their formal definitions.

\begin{definition}[Computation probability and output probability]
Let $t:T$ be any term, $\calt_t$ be the reduction tree of $t$, $l$ be any leaf of $\calt_t$, and $o$ be any term occurring at a leaf of $\calt_t$. The probability $\probab_{\mathrm{C}}(l, \calt_t)$ of the computation of $t$ yielding $l$ is defined as the product of all labels occurring on the path from the root of $\calt_t$ to $l$. The probability $\probab (o, \calt_t)$ of the potential output $o$ of $t$ is the sum $\sum_{i=1}^{n} \probab(l_i)$ where $l_1 ,\ldots , l_n$ is an enumeration of all leaves of $\calt_t$ at which the term $o$ occurs.
\end{definition}

\begin{definition}[Confidence value]\label{def:confidence}
For any typed \lamt-term $t:T$, the confidence value $\conf(t, \calp)$ of $t$ with respect to the distribution $\calp$ is defined as follows: 
\[\conf (t, \calp)= f\]where $f$ is the function such that, for every natural number $n$,
\[ f(n)=\probab (\tert, \calt_{\trust (\test_n t) (\calp)(\frac{1}{20})})\]
\end{definition}
In other words, $\conf(t, \calp)$ is the function that associates to every natural number $n$ the probability that the trustworthiness check with respect to the distribution $\calp$ with threshold $\frac{1}{20}$ conducted on an experiment consisting of $n$ executions of $t$ yields a positive result. We fix here the threshold to $\frac{1}{20}=0,05$ since it is a standard choice in the field of statistics. This number is usually set to $0,05$ or to $0,01$, depending on the phenomena considered by the particular statistical investigation at hand, and it is often called either {\it level of significance} or P-value. Obviously, nevertheless, our definition can be adapted with no effort to other thresholds.

It is easy to define then partial orders $\prec, \preceq$ and an equivalence relation $\equiv$ on the confidence values $\conf(t,\calp)=f$ by comparing how fast the values of $f(n)$ grow for $n$ approaching infinity. 

\begin{definition}[Confidence order]\label{def:confidence-order}
For any two terms $t,t'$, probability distributions $\calp, \calp'$ and functions from the natural numbers to the real numbers $f$ and $g$ such that $ \conf(t,\calp) =f$ and $\conf(t',\calp')=g$, we have 
\[ \conf(t,\calp)\prec \conf(t',\calp')\qquad 
\Leftrightarrow 
\qquad\lim_{n\rightarrow\infty}\frac{f(n)}{g(n)}<1\]
\[ \conf(t,\calp)\equiv\conf(t',\calp')\qquad 
\Leftrightarrow 
\qquad \lim_{n\rightarrow\infty}\frac{f(n)}{g(n)}= 1\]
\end{definition}


\begin{proposition}The relation $\prec$ is a strict partial order.\end{proposition}
\begin{proof}We prove that $\prec$ is irreflexive, asymmetric  and transitive.
\begin{itemize}
\item Irreflexivity. Indeed, $\lim_{n\rightarrow\infty}\frac{f(n)}{f(n)}=1\ngtr 1$.  
\item Asymmetry. Indeed, if $\lim_{n\rightarrow\infty}\frac{f(n)}{g(n)}=m < 1$ then $\lim_{n\rightarrow\infty}\frac{g(n)}{f(n)}=\frac{1}{m}>1$.
\item Transitivity. Indeed, $\frac{f(n)}{h(n)} =\frac{f(n)}{g(n)} \cdot \frac{g(n)}{h(n)} $ and if $ \lim_{n\rightarrow\infty}\frac{f(n)}{g(n)}<1$ and $  \lim_{n\rightarrow\infty}\frac{g(n)}{h(n)}<1$, we have also that $  \lim_{n\rightarrow\infty}\frac{f(n)}{h(n)}<1$ since the limit of the product of two functions equals the product of their limits.
\end{itemize}
\end{proof}

\begin{proposition}The relation $\equiv$ is an equivalence relation.
\end{proposition}
\begin{proof}We prove that $\equiv$ is reflexive, symmetric  and transitive.
\begin{itemize}
\item Reflexivity. Indeed,  $\lim_{n\rightarrow\infty}\frac{f(n)}{f(n)}= 1$. 
\item Symmetry. Indeed, if $\lim_{n\rightarrow\infty}\frac{f(n)}{g(n)}=1$ then $\lim_{n\rightarrow\infty}\frac{1}{\left(\frac{f(n)}{g(n)}\right)}=\frac{1}{1}=1$ since the limit of the constant function $1$ is $1$ and the limit of the division of two functions is the division of their limits. But $\frac{1}{\left(\frac{f(n)}{g(n)}\right)}=\frac{g(n)}{f(n)}$ and, therefore, $\lim_{n\rightarrow\infty}\frac{g(n)}{f(n)}=1$ as well.

\item Transitivity. Indeed, $\frac{f(n)}{h(n)} =\frac{f(n)}{g(n)} \cdot \frac{g(n)}{h(n)} $ and if $ \lim_{n\rightarrow\infty}\frac{f(n)}{g(n)}=1$ and $  \lim_{n\rightarrow\infty}\frac{g(n)}{h(n)}= 1$, we have also that $  \lim_{n\rightarrow\infty}\frac{f(n)}{h(n)}= 1$ since the limit of the product of two functions equals the product of their limits.
\end{itemize}
\end{proof}


\begin{theorem}[Increasing trust]\label{thm:increasing-trust}
Let $t:T$ be any term, $o_1, \ldots , o_n$ be an enumeration without repetitions of the terms occurring at the leaves of $\calt _t$, $\calp$ be any probability distribution. If \begin{itemize}\item $o_i:T_i$ for any $1\leq i\leq m$ \item $T_i\neq T_j$ for any $1\leq i\neq j\leq m$ \item $\calp = (\probab(o_1,\calt_t)\,T_1, \ldots , \probab(o_m,\calt_t)\, T_m)$\end{itemize}then $f(n)=\conf(t, \calp)$ approaches $1$ with probability $1$, for $n$ approaching infinity.
\end{theorem}
\begin{proof}
This theorem is a rather direct application of the strong law of large numbers -- see, for instance, \cite[Sec. 6.6, Thm. 7]{ash08}.

Let us consider a term $t$ with potential outputs $o_i$, for $1\leq i\leq m$, and let us define $m$ random variables as follows: each $X^i$, for $1\leq i\leq m$, has possible values $0$ and $1$ and  expected value $E(X^i)= (1\cdot \probab(o_i,\calt_t))+(0\cdot (1-\probab(o_i,\calt_t)))=\probab(o_i,\calt_t)$. A sequence of random variables $X_1^i, \ldots , X_n^i$ then encodes the count of how many times the output $o_i$ is obtained during an experiment consisting of $n$ executions of the term $t$. The expected value $E(X^i)= \probab(o_i,\calt_t)$ indeed corresponds to the frequency with which we expect $o_i$ to be obtained during an experiment on $t$. 

Notice that a sequence $X_1^i, \ldots , X_n^i$ such that $\frac{X_1^i, \ldots , X_n^i}{n}=E(X^i)$ encodes a count of the number of outcomes $o_i$ during an experiment consisting of $n$ executions of $t$ such that $ \probab(o_i,\calt_t)$ of the total outcomes was exactly $o_i$. Which means that such a sequence $X_1^i, \ldots , X_n^i$, for which $\frac{X_1^i, \ldots , X_n^i}{n}=E(X^i)$ holds, corresponds to tuples 
$\lan t_1, \ldots , t_n\ran$ of outputs of $t$ such that the number of occurrences of $o_i$ in the tuples is exactly $ \probab(o_i,\calt_t)$ of the total. In other words, since $o_i$ is the only output of $t$ of type $T_i$, such a sequence exactly corresponds to the tuples encoding a frequency distribution that perfectly agrees with the probability distribution $\calp$ with respect to the probability of outputs of type $T_i$. More in general, given any sequence $X_1^i, \ldots , X_n^i$, the closer $\frac{X_1^i, \ldots , X_n^i}{n}$ is to $E(X^i)$, the closer the frequency distribution encoded by the tuples corresponding to $X_1^i, \ldots , X_n^i$ is to the probability distribution $\calp$ with respect to the probability of $T_i$.

In case all random variables $X_1^i, \ldots , X_n^i$ have the same expected value $E(X^i)$, the strong law of large numbers states that 
\[\lim _{n \rightarrow \infty} \left(\frac{X_1^i + \ldots + X_n^i}{n}   \right)=E(X^i) \]with probability $1$. And this clearly holds for any $1\leq i\leq m$.

Since each sequence $X^1_i, \ldots , X^n_i$, for $1\leq i\leq m$, such that $\frac{X^1_i, \ldots , X^n_i}{n}=E(X_i)$ exactly corresponds to the tuples encoding a frequency distributions that perfectly agrees with the probability distribution $\calp$ with respect to the probability of outputs of type $T_i$ and since the closer $\frac{X_1^i, \ldots , X_n^i}{n}$ is to $E(X^i)$, the closer the  frequency distribution encoded by the tuples corresponding to $X_1^i, \ldots , X_n^i$ is to the probability distribution $\calp$ with respect to the probability of $T_i$; the strong law of large numbers  implies that, with probability $1$, for $n$ approaching infinity, the tuples $\lan t_1, \ldots , t_n\ran$ of outputs of $t$ that we obtain encode frequency distributions which are closer and closer to the probability distribution $\calp$ with respect to all potential outputs $o_1:T_1, \ldots , o_m:T_m$ of $t$.


If $n$ approaches infinity, $\probab (\tert, \calt_{\trust (\test_n t) (\calp)(\frac{1}{20})})$ approaches $1$ with probability $1$. Indeed, $\probab (\tert, \calt_{\trust (\test_n t) (\calp)(\frac{1}{20})})$ is the probability of obtaining $\tert$ in an experiment consisting of $n$ executions of $t$, where $\tert$ is obtained when the result of the experiment encodes a frequency distribution with distance not greater than $\frac{1}{20}$ from $\calp$. But since $\probab (\tert, \calt_{\trust (\test_n t) (\calp)(\frac{1}{20})})$ approaches $1$, by Definition \ref{def:confidence},  also $f(n)=\conf(t, \calp)$ approaches $1$ with probability $1$.
\end{proof}
Notice that the statement can be easily generalised to the slightly more complex case in which some distinct potential outputs $o_i$ and $o_j$ have the same type $T_i$ by requiring that $\calp = (\probab(o_1,\calt _t)\,T_1, \ldots , \probab(o_1,\calt_t)+\probab (o_j, \calt_t)\,T_i ,\ldots  , \probab(o_n,\calt_t)\, T_n)$.

\section{Conditional probability, conjunctions, disjunctions}\label{sec:conditional-probability}

As shown in Section \ref{sec:potential-analysis}, the calculus \lamt{} allows to predict the exact probability of each potential output of any term by studying its reduction tree. In the following, we show how to easily encode conditional constructs using as inputs probabilistic terms that behave as expected with respect to this method for computing the probability of potential outputs. In the process, we show that the calculus supports a straightforward way of computing the probability of conjunctions of probabilistic events. Afterwards, we briefly show that the probability of disjunctions of probabilistic events is also easily computed.

Let us begin with conditional constructs. If we consider a probabilistic choice $t=\{p_1t_1, \ldots , p_nt_n\}$ and several terms $s_1, \ldots , s_n$, we can easily construct a conditional program $t'$ that selects a term $s_i$ from $s_1, \ldots , s_n$ on the basis of the probabilistic output of $t$ and then behaves just like $s_i$. If  the terms $s_1, \ldots , s_n$ are probabilistic terms themselves, we have that the output of $t'$ probabilistically depends on the outputs of $s_1, \ldots , s_n$ conditioned with respect to the output of $t$.

\begin{example}
    
The output of an unfair coin $\{\frac{1}{3} \hcoin , \frac{2}{3}\tcoin\}$ can be used to trigger two different probabilistic processes $\{\frac{1}{3} 1 , \frac{1}{3}2, \frac{1}{3}3\} $ and $\{\frac{1}{3} 4 , \frac{1}{3}5, \frac{1}{3}6\}$ in such a way that the first is run if, and only if, the coin yields $\hcoin$ and the second is run if, and only if, the coin yields $\tcoin$. We can do this in a modular way by, first, defining the abstraction $\lambda x_1 .\lambda x_2 . \{\frac{1}{3} \lan \hcoin , x_1 \ran , \frac{2}{3}\lan \tcoin , x_2 \ran\}$ of the term implementing the unfair coin. An abstraction of this kind can be simply obtained, in general, by manually replacing the terms of the choice -- $\hcoin, \tcoin$ in our case -- by pairs containing one of the original terms and one of the bound variables -- $x_1,x_2$ here. Secondly, we apply this abstraction to the  terms that are supposed to depend on the output of the choice: \[(\lambda x_1 .\lambda x_2 . \{\frac{1}{3} \lan \hcoin , x_1 \ran , \frac{2}{3}\lan \tcoin , x_2 \ran\})\; \{\frac{1}{3} 1 , \frac{1}{3}2, \frac{1}{3}3\} \;\{\frac{1}{3} 4 , \frac{1}{3}5, \frac{1}{3}6\}\] The reduction of this term will yield the probabilistic term \[\{\frac{1}{3}  \lan \hcoin , \{\frac{1}{3} 1 , \frac{1}{3}2, \frac{1}{3}3\}\ran  , \frac{2}{3}  \lan \tcoin , \{\frac{1}{3} 4 , \frac{1}{3}5, \frac{1}{3}6\}\ran \}\] which reduces to $ \lan \hcoin ,  \{\frac{1}{3} 1 , \frac{1}{3}2, \frac{1}{3}3\}\ran  $ exactly when our original unfair coin would have yielded $\hcoin $ and to $\lan \tcoin , \{\frac{1}{3} 4 , \frac{1}{3}5, \frac{1}{3}6\}\ran $ exactly when our unfair coin would have yielded $\tcoin$.
\end{example}

Let us formalise this intuition and define a general procedure to build conditional constructs of this kind.

\begin{definition}[Conditional construct]\label{def:conditional-construct}
For any term $t=\{p_1t_1, \ldots , p_nt_n\}$ and list of terms $s_1, \ldots , s_n$ , we denote by \[(s_1, \ldots , s_n\mid t) \] the term \[(\lambda x_1  .\, \ldots \,  \lambda x_n. \{p_1\lan t_1, x_1\ran, \ldots , p_n\lan t_n , x_n\ran \}) s_1\, \ldots\, s_n\]
\end{definition}

As for the typing of the tuples used in the previous definition, notice that, for any list of terms $u_1:U_1, \ldots , u_m:U_m$, we can construct a type derivation for the term $\lan u_1 , \ldots , u_m \ran: (U_1 \oplus \ldots \oplus U_m)^m $ since, by Definition \ref{def:subtype}, each type in the list $U_1 , \ldots , U_m$ is a subtype of $U_1 \oplus \ldots \oplus U_m$.

We can now study what kind of relationship conditional constructs of this kind entertain with conditional probability.
\begin{definition}[Conditional probability]\label{def:conditional-probability} Let \[c=(s_1, \ldots , s_n \mid \{p_1t_1, \ldots , p_nt_n\})\]be a conditional construct. 
The conditional probability $\probab (u\mid t_i, \calt_{c} )$, for $1\leq i\leq n$ -- that is, the probability of obtaining $u$ as second element of the output of $c$ conditioned on the fact that the first element of the output is $t_i$ -- is defined as 
$\probab (u , \calt_{s_i})$.

\end{definition}


%
%

We show that the numerical value that we obtain by this definition of conditional probability is exactly the same that we would obtain by the usual definition of  $\mathrm{p}(u\mid t_i)$ -- notice that we use $\probab $ for the probability function that we have just defined and $\mathrm{p}$ for the traditional notion of probability defined by probability theory. If we follow the usual definition of conditional probability, we have that  $\mathrm{p}(u\mid t_i)=\frac{\mathrm{p}(u \wedge t_i)}{\mathrm{p}(t_i)}$, that is, the probability of obtaining $u$ given that we obtained $t_i$ equals the probability of obtaining $u$ together with $t_i$ divided by the probability of obtaining $t_i$. Now, the probability of obtaining $u$ together with $t_i$ is simply the probability of obtaining $\lan t_i , u\ran$ as output of our process, and the probability of obtaining $t_i$ is the probability of obtaining any output of the form $\lan t_i, \ldots \ran $. Hence, the following proposition  proves that our definition of conditional probability $\probab (u\mid t_i , \calt_{(s_1, \ldots , s_n \mid \{p_1t_1, \ldots , p_nt_n\})})$ exactly corresponds to the usual definition of  conditional probability $\mathrm{p}(u\mid t_i)$.  

\begin{proposition}
For any conditional construct \[c=(s_1, \ldots , s_n \mid \{p_1t_1, \ldots , p_nt_n\})\]the  probability $\probab (u\mid t_i, \calt_c)$ that $c$ yields  any pair of the form $\lan \ldots  , u\ran $ conditioned on the fact that $t_i$ is the first element of the pair equals the probability that $c$  yields exactly $\lan t_i , u\ran$ divided by the probability that $c$  yields any term of the form $\lan t_i , \ldots \ran$.
\end{proposition}
\begin{proof}
By definition, $\probab (u\mid t_i, \calt_c) =\probab (u , \calt_{s_i})$. The probability of obtaining $\lan t_i , u\ran$ as output of $c$, on the other hand, is $\probab (\lan t_i , u\ran, \calt_c)$. This value is computed by taking the product of the labels on the path connecting the root of $\calt_c$ to its leaf of the form $\lan t_i , u\ran$.
Let us define the {\it upper part} of this path to be its shortest subpath connecting the root of $\calt_c$ to a node of the form $\lan t_i , \ldots \ran$ and let us call this particular node $\lan t_i , \ldots \ran$ the {\it mid-node} of the path. The {\it lower part} of the path is then defined as the subpath that connects its mid-node to the leaf $\lan t_i , u \ran$. Notice that the upper and the lower part of the path meet at the mid-node and their union is the whole path.

Now, the product of the labels on the upper part of our path is exactly the probability that $c$ yields a term of the form $\lan t_i , \ldots \ran$. Indeed, this is the product of the labels on the path connecting the root of $\calt_c$ to the mid-node, and a leaf of $\calt_c$  is of the form $\lan t_i , \ldots \ran$ if, and only if, it occurs in the subtree rooted at the mid-node.

For obvious arithmetical reasons, if we divide $\probab (\lan t_i , u\ran, \calt_c)$ by the product of the labels occurring on the upper part of our path, we obtain the product of the labels occurring on the lower part of our path. But the product of the labels occurring on the lower part of our path is exactly $\probab (\lan t_i ,  u\ran , \calt_{s_i})$. Indeed, the value of $ \probab (\lan t_i ,  u\ran , \calt_{s_i}) $ is computed by taking the product of the labels occurring on the path connecting the root of $\calt_{s_i}$ to the leaf $\lan t_i ,  u\ran$ and, by the definition of conditional constructs and of reductions trees, the labels occurring in $\calt_{s_i}$ are identical to the labels occurring in the subtree of $\calt_c$ that is rooted at the mid-node.
\end{proof}

From the previous proof, it is easy to see that the calculus also enables us to encode conjunctions of probabilistic events -- intended as outputs of probabilistic  terms -- and to compute their probability exactly as one would expect. Indeed, for any two  probabilistic terms $t:T$ and $s:S$, we can construct the term $\lan t,s\ran: T\oplus S $. Then, if $t$ yields $u$ as output with probability $p$ and $s$ yields $v$ as output with probability $q$, then $\lan t,s\ran$ will yield  $\lan u,v\ran $ as output with probability $p\cdot q$. Accordingly, we have that if $\probab(u, \calt_{t})=p$ and $\probab(v, \calt_{s})=q$, then, by the definition of reduction trees, $\probab(\lan u, v\ran , \calt_{\lan t, s\ran })=p\cdot q$. 

As for the probability of the disjunction of events, since the interesting case is the one in which we consider more outputs of the same process, we must do something slightly different and exploit the operators $\trust $ and $\test$ in a creative way. Suppose that we have a process $t$ with potential outputs $u_1:U_1, \ldots , u_n:U_n$ and we want to compute the probability of obtaining either $u_i:U_i$ or $u_j:U_j$ for $1\leq i,j\leq n$ and  $U_i\neq U_j$. In this case, we can construct the term $\trust \, (\test_1 t) \, (1\; U_i\oplus U_j)\,0 $ and then compute $\probab(\tert, \calt_{\trust \, (\test_1 t)\, (1\; U_i\oplus U_j)\, 0})$. Indeed, this value will be the probability of obtaining in one execution of $t$ either one output of type $U_i$ or one of type $U_j$. Obviously, this probability will precisely be the sum of the probability of obtaining $u_i$ and of the probability of obtaining $u_j$, as one would expect.

As final remark, let us specify that the notion of conditional  construct specified in Definition \ref{def:conditional-construct}
also covers the case in which our conditional process is a term $(s_1, \ldots , s_n  \mid \{p_1t_1, \ldots , p_nt_n\})$ where some of the choices $t_1, \ldots , t_n$ -- the list of terms potentially obtained during the first phase of the computation -- are not only used to select how to continue the computation but are also used as arguments of the processes $s_1, \ldots , s_n$ implementing the second phase of the computation. It is indeed enough to consider the conditional construct\[((\lam x_1.r_1)t_1, \ldots , (\lam x_n.r_n)t_n \mid \{p_1t_1, \ldots , p_nt_n\})\]
The existence of this particular case of our notion of conditional construct is worth mentioning since a natural application of conditional probability relies on it. Suppose indeed that we wish to represent the probability that a process $t$ yields a particular output $v$ in case a subprocess $s$ of $t$ yields a particular output $s'$. More precisely, we would like to have a formal means to represent the conditional probability statement that expresses the probability that the term $t[s/x]$ reduces to the value $v$ under the assumption that its subterm $s$ reduces to the value $s'$.\footnote{The substitution $t[s/x]$ is simply used here to specify that the term $s$ occurs inside the term $t$.} In order to do so, we can define the conditional construct \[((\lam x . t)s_1 , \ldots ,(\lam x . t)s_n  \mid \{p_1s_1 , \ldots , p_ns_n\})\]where $s_1 , \ldots , s_n$ are all values to which $s$ might reduce and $p_1 , \ldots , p_n$ are their respective probabilities of being the actual value to which $s$ reduces, and then define the conditional probability statement \[\probab(v \mid s' , \calt_{((\lam x . t)s_1 , \ldots ,(\lam x . t)s_n  \mid \{p_1s_1 , \ldots , p_ns_n\})})\]in which, for some $i$ such that  $1\leq i\leq n$, $s_i=s'$. If we do so, we have that, by definition\[\probab(v \mid s' , \calt_{((\lam x . t)s_1 , \ldots ,(\lam x . t)s_n  \mid \{p_1s_1 , \ldots , p_ns_n\})})=\probab (v , \calt_{(\lam x . t)s'})\]and $\probab (v , \calt_{(\lam x . t)s'})$ is obviously equal to  the probability that $t[s'/x]$ reduces to $v$, which is what we wished to represent.\footnote{This is exactly the notion of conditional probability formalised by TPTND judgements of the form $\ldots , x_t:\alpha_a , \ldots \vdash t_n : \beta_{\tilde{b}} $ or of the form $\ldots \vdash [x_t]t_n : (\alpha \impl \beta)_{[a]\tilde{b}} $. Judgements of this kind, indeed, express the fact that $n$ executions of the process $t$ show that $t$ yields $\beta$ as output with probability $b$ given that the random variable $x_t$, on which $t$ depends, takes  value $\alpha$ with probability $a$.}


\section{Related work}\label{sec:related}

Much work has been and is presently being devoted, in the context of the literature on logic and computation, to formal analyses of probabilistic processes and phenomena through deductive and computational systems. Three main approaches can be distinguished: a purely deductive one, a purely computational one, and a mixed one concerned with the definition of typed calculi based on correspondences between deductive systems and computational ones. The deductive approaches mainly aim at relativising consequence relations by using probabilities. They thus formalise relations that hold when something follows from certain hypotheses with a certain probability. The purely computational approaches aim, on the other hand, at relativising the notion of computation itself by probabilities, and they formalise relations that hold when something is the output of a certain process with a certain probability. The third approach oscillates between the first two and either only adds a deductive component for assigning types to the terms of a probabilistic computational system or does this and, in addition, presents a deductive component that can be interpreted as a probabilistic relativisation of a consequence relation.    

Even though in Sections \ref{sec:potential-analysis} and \ref{sec:conditional-probability} we also introduce methods for reasoning about trust-related properties of terms and for conducting logical operations on terms, the calculus \lamt{} in itself is a system of the third kind where the deductive component is simply aimed at providing a useful typing system for the computational terms of the calculus. The definition of introduction and elimination rules for a trustworthiness operator, defined on the basis of the operators $\test$ and $\trust$, seems possible and of sure interest. Extending the calculus in this direction would yield a logical analysis of trust in addition to the computational one presently provided. Considering that the presentation of \lamt{} alone already constitutes a considerable amont of material, though, we prefer to leave such an extension for future work.

We present a brief overview of the existing systems 
for the formalisation of probabilistic processes and phenomena through deductive and computational systems in order to discuss the relations and differences between these and \lamt{}.

Let us begin with those systems that admit a clear deductive reading of the analysed phenomena. A system of this kind is presented and studied in \cite{bor19}. The work  introduces a sequent calculus in which each sequent -- which can be seen as a derivability statement -- is labelled by a probability interval of its validity. The properties of the deductive rules for working with probabilistic statements of this kind are then studied, while the work does not offer a computational interpretation for the derivations that can be constructed by the introduced deductive system.

A system which, on the other hand, implements a deductive approach to the analysis of probabilistic computational processes is TPTND \cite{pda21,dgp25}. While TPTND and  \lamt{} share the same goal -- that is, the formalisation of a notion of trustworthiness  of probabilistic processes -- there is a significant difference between the two concerning the level of the conducted analysis.  TPTND is a high-level, descriptive deductive system in which probabilistic processes are studied by assuming a -- possibly incomplete -- probabilistic description of their behaviour with respect to their potential outputs. In \lamt{}, on the other hand, a probabilistic process can be directly internalised in the calculus as a $\lam$-term that computationally implements the behaviour of the process. The probabilistic description of this behaviour is then extracted at runtime in order to conduct the analysis of the process in terms of trustworthiness. We, therefore,  have in  \lamt{} a lower-level analysis of the processes that does not assume a pre-given probabilistic description. Another difference between TPTND and \lamt{} is that even though the former is mainly meant to provide a way to analyse probabilistic processes from a computational perspective, it also supports, unlike \lamt{}, a deductive declarative reading.  The probabilistic component of TPTND is used to specify with what probability a certain term (intended as a computational process) has a certain type (intended as one of its possible outputs): the notation $\Gamma\vdash t:\alpha_{p}$ is most naturally interpreted as stating that the process $t$ has probability $p$ to yield $\alpha$ as output under the probability distribution encoded in $\Gamma$. But the fact that a term has a certain type with a certain probability can also be naturally interpreted as the fact that the term $t$ witnesses that the occurrence of the event $\alpha $ follows from the assumptions $\Gamma$ with a certain probability.\footnote{The naturalness of this interpretation derives from the fact that in certain approaches to type theory, see \cite{ml84}, and in proofs-as-programs correspondences, see for instance \cite{glt89},  $\Gamma\vdash t:\alpha$ is also interpreted as stating that $t$ witnesses that $\alpha $ follows from the formulae occurring in $\Gamma$.} Statements of the form $\Gamma\vdash t:\alpha_{p}$, therefore, also have a rather natural declarative interpretation. Such an interpretation is not presently supported by \lamt{} since probabilities are never explicitly used to relativise the relation between terms and types. On the other hand, \lamt{} enables us to explicitly represent the structure of probabilistic programs, which is left implicit in TPTND. This last difference between the two systems also explains why they induce two rather different formal notions of computation. Indeed, normalisation in TPTND does not concern in any way the {\it reduction} of terms -- intended as simplification of terms' structure -- but it rather formalises the  process of conducting probabilistic experiments and of combining their outcomes through the process of applying rules to previously obtained derivations. This is a rather non-standard notion of normalisation, but it is justified by the  fact that  experiments on the probabilistic behaviour of processes are executed in TPTND precisely by applying deduction rules. Normalisation in \lamt{}, on the other hand,  consists in the reduction of terms to redex-free normal forms or to values, as usual. This also means that, unlike TPTND, \lamt{} provides us with an explicit procedural analysis of probabilistic processes.

A considerable literature concerns $\lam$-calculus-based systems with no deductive reading but providing instead a purely computational analysis of probabilistic processes. In \cite{dp95}, for instance, a general discussion of non-deterministic extensions of $\lam$-calculus can be found. No specific reference to probability is made, but the studied issues are quite typical of probabilistic $\lam$-calculi and a very clear distinction between parallelism and non-determinism -- the latter being a characterising feature of probabilistic computation -- is traced.
Several specific programming  languages based on $\lam$-calculus but also featuring probabilistic features have been introduced and studied in the literature, see for instance  \cite{rp02}, \cite{dlz12}, \cite{sco14}, \cite{dk20}, and \cite{akm21}. We remark that a semantical approach has been usually favoured. A special mention to the work in \cite{dhw05} is due, though, since one of the main ideas behind the term-based semantics developed in \cite{dhw05} is at the basis of the kind of analysis developed in Section \ref{sec:potential-analysis} -- the notion of {\it probabilistic transition relation derivation} is indeed comparable with the notion of reduction tree introduced by  Definition \ref{def:reduction-tree}.

Finally, several systems combining a computational analysis and a deductive one have been introduced. Since their number is more limited and since these calculi are closer to \lamt{} in nature, we discuss them in more detail. A particularly expressive probabilistic calculus has been presented in \cite{par03}. This calculus is based on sampling functions and elementary probabilistic instructions in this calculus are limited to the generation of random real numbers: $\gamma . t $, with $t$ term and $\gamma$ sampling variable, reduces to $t[u/\gamma ]$\footnote{Notice that the expression $t[u/\gamma ]$ uses our notation for substitutions. The notation employed in  \cite{par03} is different.} where $u$ is a randomly generated number. Nevertheless, it is clear that this kind of instruction is enough to simulate probabilistic choices among generic terms. The calculus also features a fixpoint operator that enables the formulation of recursive definitions. Hence, in this calculus it is also possible to specify probability distributions over infinite data structures. Another central peculiarity of this system is that it features a syntactic distinction between regular functions and functions the application of which can only be reduced when applied to values. Indeed, in the calculus, applications of regular lambda abstractions $\lam x:\tau .t $ are reduced according to a call-by-value discipline, while applications of integral abstractions $\int x:\tau . t $ are reduced according to a call-by-name discipline. Thus, randomly generated numbers can be fed to terms according to a call-by-name discipline while other arguments can still be used according to a call-by-value discipline. 

A second rather expressive probabilistic language obtained by extending $\lam$-calculus has been introduced in \cite{ppt05}. The basic probabilistic instruction in the calculus is the sampling of a distribution and, just like the calculus presented in \cite{par03}, it is based on sampling functions. The main difference between the calculus introduced in \cite{par03} and the one introduced in \cite{ppt05} is that the syntax of the latter -- which defines a metalanguage in which two languages can be distinguished: one for terms and one for expressions -- features a direct way to denote probability distributions. This makes it possible to have a basic probabilistic operator that takes as explicit arguments one probability distribution $P$, one term $t$ and one variable $x$ occurring in $t$, and outputs an instantiation of the term $t$ in which a value sampled from the probability distribution replaces all occurrences of $x$ in $t$. This enables a selective way of using call-by-value and call-by-name reduction strategies similar to the one adopted in \cite{par03}. Also this calculus features a fixpoint operator which enables the encoding of infinite data structures and the direct formulation of recursive definitions. 
Both calculi presented in \cite{ppt05} and \cite{par03} distinguish themselves from similar calculi in the literature insofar as their expressive power allows to represent continuous distributions. The kind of probabilistic choice featured in \lamt{} -- that is, the one formalised by terms of the form $\{p_1t_1, \ldots p_nt_n\}$ -- does not enable us to represent distributions of this kind.  


The calculus  $\Lambda _{\mathsf{HP}}$ introduced \cite{te19} is a probabilistic $\lam$-calculus equipped with a linear logic type system. Apart from the obvious difference that the types of $\Lambda _{\mathsf{HP}}$-terms are linear logic formulae while the types of \lamt{}-terms are not, and the fact that only binary non-deterministic choices are primitive in $\Lambda _{\mathsf{HP}}$ -- which is nonetheless not a significant difference, since multiple choices can be simulated by binary ones -- a substantial difference between $\Lambda _{\mathsf{HP}}$ and \lamt{} consists in the fact that the main strategy adopted in $\Lambda _{\mathsf{HP}}$ for the evaluation of probabilistic arguments of functions is call-by-value, as opposed to \lamt{} which employs a call-by-name strategy. The authors of this work lean towards the adoption of a call-by-value strategy for the evaluation of function arguments in order to be able to fix the value of probabilistic terms before they are duplicated and used by other functions; in the present work, we do the opposite in order to be able to conduct experiments testing different outputs of the same probabilistic program, and hence exactly in order {\it not to fix} the value of probabilistic terms before they are duplicated to be used inside the scope of operators such as $\test $ and $\trust $. 

As far as expressive power is concerned, all calculi presented in \cite{ppt05}, \cite{par03} and \cite{te19} are more expressive than the $\test$-$\trust$-free fragment of \lamt{}. Indeed, recursive datatypes are supported in $\Lambda _{\mathsf{HP}}$ and non-terminating terms can be defined. Similarly, by the calculi presented in \cite{ppt05} and \cite{par03}, it is possible to define infinite data structures and directly encode recursive definitions, which is not possible in \lamt{}. On the other hand, \lamt{} enjoys a termination result, both for the full reduction of terms to normal form without redexes by a logical reduction strategy and for the reduction of terms to values by a call-by-name strategy. Moreover, the calculus \lamt{}, unlike those introduced in \cite{ppt05}, \cite{par03} and \cite{te19}, features the  operators $\test$ and $\trust $ for conducting experiments and trustworthiness checks on probabilistic programs. No operator that can be directly used to do so appears in \cite{par03} or \cite{te19}. As for \cite{ppt05}, on the other hand, it seems possible to use the $\mathsf{expectation}$ term construct of \cite{ppt05} for tasks which are similar to those that can be achieved by $\test$ and $\trust$ in \lamt{}. The main difference seems to be that while the $\mathsf{expectation}$ operator of \cite{ppt05} takes as argument a term explicitly denoting the considered probability distribution, $\test$ and $\trust$ in \lamt{} do not require in any way knowledge about the probability distribution that defines the behaviour of the considered probabilistic program. It is quite clear, indeed, that $\test$ and $\trust$ can be very well also used to analyse the probabilistic behaviour of blackbox programs of unknown structure by exclusively considering their outputs.

A typed probabilistic $\lam$-calculus which is, on the other hand, closer in expressive power to the $\test$-$\trust$-free fragment of \lamt{} is the one introduced in \cite{dip20}. The latter constitutes indeed one of the main inspirations of the  $\test$-$\trust$-free fragment of our calculus. The typing policy adopted for probabilistic choices and the way function applications are handled during reduction in \cite{dip20}, nevertheless, essentially differ from those employed in the present work. The main difference in typing concerns the definition of subtypes. Di Pierro defines $\cho _{i=1}^{m} A_j$ to be a subtype of $\cho _{i=1}^{n} B_i$ if, and only if, for each $B_i$, there is an $A_j$ which is a subtype of $B_i$. For us, on the other hand, $\cho _{i=1}^{m} A_j$ is a subtype of $\cho _{i=1}^{n} B_i$ if, and only if, each $A_j$ is a subtype of some $B_i$. Notice the inversion of the quantified elements in the two definitions. The intuition behind our definition is directly tied to the idea that a subtype of a type is a specification of the type, in the sense that the subtype should be usable in all cases in which the type is usable. And indeed we can use $\cho _{i=1}^{m} A_j$ whenever we can use $\cho _{i=1}^{n} B_i$, since each element of $\cho _{i=1}^{m} A_j$ is a specification of some element of $\cho _{i=1}^{n} B_i$. The subtyping discipline adopted by Di Pierro, on the other hand, complies with the intuition behind the notion of subtypes insofar as a particular reduction policy is adopted for applied functions of type $\cho _{i=1}^{n} C_i\impl D_i$ -- and this is the second relevant difference between the present work and Di Pierro's. When a function term of type $\cho _{i=1}^{n} C_i\impl D_i$ is   applied to an argument, the reduction policy employed in \cite{dip20} admits the possibility of eliminating certain choices from the term of type $\cho _{i=1}^{n} C_i\impl D_i$ -- thus reducing the type itself as well -- if these choices are not compatible with the argument to which the term is applied.

\section{Conclusion}\label{sec:conclusion}

We have introduced \lamt{}, a typed probabilistic $\lam$-calculus featuring an experiment operator $\test $ and a trustworthiness operator $\trust$ that allows to conduct experiments on the probabilistic behaviour of programs and to use the results of these experiments at runtime in order to conclude whether the tested programs are trustworthy with respect to a given probability distribution of reference modulo an admissible error threshold.  We also defined a notion of confidence based on this runtime notion of trust that allows us to abstract from particular computations and to reason about the global computational behaviour of probabilistic terms. This notion of confidence induces a partial order among terms and, by an application of the strong law of large numbers, we have proved that it behaves as expected with respect to terms that are suitably defined with respect to the target probability distribution.

We have left two main investigation points for future work. The first is the formal specification of a way to integrate in the system blackbox probabilistic processes in order to make explicit the fact that the runtime behaviour of $\test$ and $\trust $ does not rely on any kind of information concerning the definition of probabilistic programs. As already mentioned, there seem to be no real obstacle to the development of such an extension of \lamt{}. The second issue that we have left for future work concerns the definition of a deductive extension of \lamt{} constituting both a computational system  and a  declarative logical system for trust. Such a system could be defined by adding to \lamt{} a group of well-balanced deduction rules for an operator that declares the trustworthiness of a program rather than checking it. These rules could be defined in accordance to a proofs-as-programs correspondence in order to provide a sentential rendition of the phenomena which are computationally analysed by \lamt{}.

\section*{Acknowledgment}
  \noindent We thank 

\bibliographystyle{apalike}
\bibliography{bib-lambda-trust.bib}

\end{document}